\documentclass{article}

\usepackage[final, nonatbib]{nips_2017}

\usepackage[utf8]{inputenc} 
\usepackage[T1]{fontenc}    
\usepackage{hyperref}       
\usepackage{url}            
\usepackage{booktabs}       
\usepackage{amsfonts}       
\usepackage{nicefrac}       
\usepackage{microtype}      

\usepackage{amsmath, amssymb, amsthm, latexsym}
\usepackage{graphicx,epsfig}
\usepackage{subfigure}
\usepackage{xspace}
\usepackage{color}
\usepackage{outlines}
\usepackage{multirow}
\usepackage[noend]{algorithmic}
\usepackage{wrapfig}

\newtheorem{theorem}{Theorem}
\newtheorem{definition}{Definition}
\newtheorem{lemma}{Lemma}

\renewcommand{\qed}{\hfill $\framebox(6,6){}$}

\renewenvironment{proof}{\par{\noindent \bf Proof:}}{\qed \par \smallskip}

\newcommand{\makeatitle}[4]{
	\newpage
	\noindent
	\begin{center}
  	\framebox[\textwidth]{
    	\vbox{
    		\vspace{4mm}
      		\hbox to 0.95\textwidth { {\bf \Large \hfill #1 \hfill} }
      		\vspace{2mm}      
      		\hbox to 0.95\textwidth { {\bf \Large \hfill Lecture #2: #3  \hfill} }
      		\vspace{2mm}      
      		\hbox to 0.95\textwidth { {\Large \hfill #4 \hfill} }
      		\vspace{4mm}
    	}
  	}
	\end{center}	
}

\newcommand{\squishlist}{
	\begin{list}{$\bullet$}{
		\setlength{\itemsep}{0pt}
		\setlength{\parsep}{3pt}
		\setlength{\topsep}{3pt}
		\setlength{\partopsep}{0pt}
		\setlength{\leftmargin}{1.5em}
		\setlength{\labelwidth}{1em}
		\setlength{\labelsep}{0.5em}
   }
}

\newcommand{\squishenum}{
	
	\begin{list}{scount}{\usecounter{scount}}{
		\setlength{\itemsep}{0pt}
		\setlength{\parsep}{3pt}
		\setlength{\topsep}{3pt}
		\setlength{\partopsep}{0pt}
		\setlength{\leftmargin}{1.5em}
		\setlength{\labelwidth}{1em}
		\setlength{\labelsep}{0.5em}
	}
}

\newcommand{\squishend}{
	\end{list}
}

\newcommand{\stitle}[1]{\vspace{0.0ex} \noindent{\bf #1}}

\newcommand{\csec}{Section~}
\newcommand{\csecs}{Sections~}

\newcommand{\cdef}{Definition~}

\newcommand{\cthm}{Theorem~}
\newcommand{\cthms}{Theorems~}

\newcommand{\cfig}{Figure~}
\newcommand{\cfigs}{Figures~}
\newcommand{\ie}{i.e.\xspace}
\newcommand{\eg}{e.g.\xspace}

\newcommand{\eat}[1]{}

\newcommand{\pr}[1]{{\bf Pr}\!\left[#1\right]\xspace}

\newcommand{\bigoh}[1]{{\rm O}\!\left(#1\right)\xspace}

\newcommand{\bigomega}[1]{{\rm \Omega}\!\left(#1\right)\xspace}

\newcommand{\eps}{{{\varepsilon}\xspace}}

\title{Collecting Telemetry Data Privately}
\vspace{-2mm}
\author{
  Bolin Ding, Janardhan Kulkarni, Sergey Yekhanin \\
  Microsoft Research \\
  \texttt{\{bolind, jakul, yekhanin\}@microsoft.com}\\
}

\begin{document}

\maketitle

\begin{abstract}
The collection and analysis of telemetry data from user's devices is routinely performed by many software companies. Telemetry collection leads to improved user experience but poses significant risks to users' privacy. Locally differentially private (LDP) algorithms have recently emerged as the main tool that allows data collectors to estimate various population statistics, while preserving privacy. The guarantees provided by such algorithms are typically very strong for a single round of telemetry collection, but degrade rapidly when telemetry is collected regularly.
In particular, existing LDP algorithms are not suitable for repeated collection of counter data such as daily app usage statistics.

In this paper, we develop new LDP mechanisms geared towards repeated collection of counter data, with formal privacy guarantees even after being executed for an arbitrarily long period of time. For two basic analytical tasks, mean estimation and histogram estimation, our LDP mechanisms for repeated data collection provide estimates with comparable or even the same accuracy as existing single-round LDP collection mechanisms. We conduct empirical evaluation on real-world counter datasets to verify our theoretical results.

Our mechanisms have been deployed by Microsoft to collect telemetry across millions of devices.
\end{abstract}


\section{Introduction}
\label{sec:intro}

Collecting telemetry data to make more informed decisions is a commonplace. In order to meet users' privacy expectations and in view of tightening privacy regulations (\eg, European GDPR law) the ability to collect telemetry data privately is paramount. Counter data, \eg, daily app or system usage statistics reported in seconds, is a common form of telemetry. In this paper we are interested in algorithms that preserve users' privacy in the face of continuous collection of counter data, are accurate, and scale to populations of millions of users.

 \begin{wrapfigure}{r}{0.35\textwidth}
 	\vspace{-0.6cm}
 	\begin{center}
 		\includegraphics[width=0.40\textwidth]{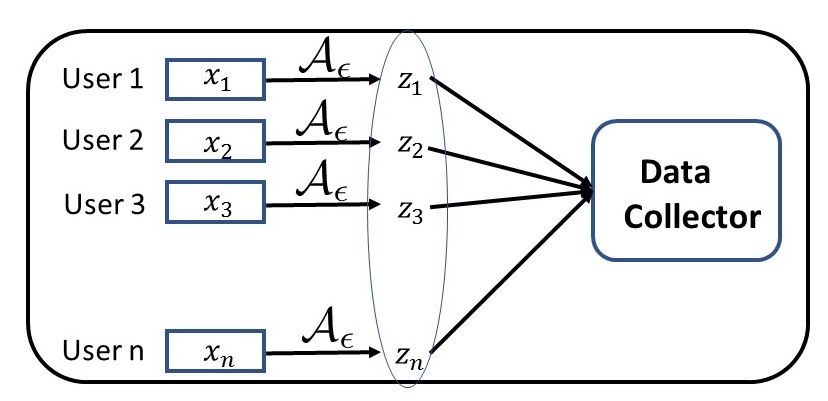}
 	\end{center}
 	\caption{Local Model of Differential Privacy}
 	\vspace{-0.5cm}
 \end{wrapfigure}

Recently, differential privacy~\cite{TCC06} (DP) has emerged as defacto standard for the privacy guarantees. In the context of telemetry collection one typically considers algorithms that exhibit differential privacy in {\em the local model}~\cite{ccs:ErlingssonPK14, pvldb:HuYYDCYGZ15, popets:FantiPE16, focs:DuchiJW13, focs:BassilyST14, NIPS:BNST017,Arxiv:TKBWW}, also called randomized response model \cite{warner1965randomized}, $\gamma$-amplification \cite{evfimievski2003limiting}, or FRAPP \cite{agrawal2005framework}. These are randomized algorithms that are invoked on user's device to turn user's private value into a response that is communicated to data collector and have the property that the likelihood of any specific algorithm's output varies little with the input, thus providing users with plausible deniability. Guarantees offered by locally differentially private algorithms, although very strong in a single round of telemetry collection, quickly degrade when data is collected over time.  This is a very challenging problem that limits the applicability of DP in many contexts.

In telemetry applications, privacy guarantees need to hold in the face of continuous data collection. Recently, in an influential paper~\cite{ccs:ErlingssonPK14} proposed a framework based on {\em memoization} to tackle this issue. Their techniques allow one to extend single round DP algorithms to continual data collection and protect users whose values stay constant or change very rarely. The key limitation of the work of~\cite{ccs:ErlingssonPK14} is that their approach cannot allow for even very small but frequent changes in users' private values, making it inappropriate for collecting counter data. In this paper, we address this limitation.

{\em We design mechanisms with formal privacy guarantees in the face of continuous collection of counter data. These guarantees are particularly strong when user's behavior remains approximately the same, varies slowly, or varies around a small number of values over the course of data collection.}

\textbf{Our results.} Our contributions are threefold.

\begin{itemize}
\item  We give simple $1$-bit response mechanisms in the local model of DP for single-round collection of counter data for mean and histogram estimation. Our mechanisms are inspired by those in \cite{warner1965randomized, nips:DuchiWJ13, focs:DuchiJW13, stoc:BassilyS15}, but allow for considerably simpler descriptions and implementations. Our experiments also demonstrate performance gains in concrete settings.

\item Our main technical contribution is a rounding technique called $\alpha$-point rounding that borrows ideas from approximation algorithms literature~\cite{goemans2002single, bansal2008improved}, and allows memoization to be applied in the context of private collection of counters while avoiding substantial losses in accuracy or privacy. We give a rigorous definition of privacy guarantees provided by our algorithms when the data is collected continuously for an arbitrarily long period of time. We also present empirical findings related to our privacy guarantees.

\item Finally, our mechanisms have been deployed by Microsoft across millions of devices starting with Windows Insiders in Windows 10 Fall Creators Update to protect users' privacy while collecting application usage statistics.
\end{itemize}


\subsection{Preliminaries and problem formulation}\label{sec:pre}

In our setup, there are $n$ {\em users}, and each user at time $t$ has a private (integer or real) counter with value $x_i(t)\in [0, m]$.  A {\em data collector} wants to collect these counter values $\{ x_i(t)\}_{i\in[n]}$ at each time stamp $t$ to do statistical analysis. For example, for the telemetry analysis, understanding the mean and the distribution of counter values (\eg, app usage) is very important to IT companies.

\stitle{Local model of differential privacy (LDP).} Users do not need to trust the data collector and require formal privacy guarantees before they are willing to communicate their values to the data collector. Hence, a more well-studied DP model \cite{TCC06, dwork2014algorithmic}, which first collects all users' data and then injects noise in the analysis step,  is not applicable in our setup.

In this work, we adopt the {\em local model of differential privacy}, where each user randomizes private data using a randomized algorithm (mechanism) $\mathcal{A}$ locally before sending it to data collector.

\begin{definition}[\cite{evfimievski2003limiting,nips:DuchiWJ13,stoc:BassilyS15}]
	A randomized algorithm $\mathcal{A}: \mathcal{V} \rightarrow \mathcal{Z}$ is $\epsilon$-locally differentially private ($\epsilon$-LDP) if for any pair of values $v, v' \in \mathcal{V}$ and any subset of output $S \subseteq \mathcal{Z},$ we have that \[\pr{\mathcal{A}(v) \in S} \leq e^\epsilon\cdot \pr{\mathcal{A}(v') \in S}.\]
\end{definition}

LDP formalizes a type of plausible deniability: no matter what output is released, it is approximately equally as likely to have come from one point $v \in \mathcal{V}$ as any other. For alternate interpretations of differential privacy within the framework of hypothesis testing we refer the reader to \cite{wasserman2010statistical,focs:DuchiJW13}.

\stitle{Statistical estimation problems.} We focus on two estimation problems in this paper.

{\em Mean estimation:} For each time stamp $t$,  the data collector wants to obtain an estimation $\hat{\sigma}(t)$ for $\sigma(\vec{x_t})=\frac{1}{n}\cdot\sum_{i\in [n]} x_i(t)$. The {\em error} of an estimation algorithm for mean is defined to be $ \max_{\vec{x_t} \in [m]^n}|\hat\sigma(\vec{x_t}) - \sigma(\vec{x_t})|$. In other words, we do worst case analysis. We  abuse notation and denote $\sigma(t)$ to mean $\sigma(\vec{x_t})$ for a fixed input $\vec{x_t}$.

{\em Histogram estimation:} Suppose the domain of counter values is partitioned into $k$ buckets (\eg, with equal widths), and a counter value $x_i(t) \in [0, m]$ can be mapped to a bucket number $v_i(t) \in [k]$. For each time stamp $t$, the data collector wants to estimate frequency of $v \in [k]:$ $h_t(v) = \frac{1}{n} \cdot |\{i : v_i(t) = v \}|$ as $\hat h_t(v)$. The {\em error} of a histogram estimation is measured by $\max_{v \in [k]}|\hat h_t(v) - h_t(v) |.$ Again, we do worst case analysis of our algorithms.

\subsection{Repeated collection and overview of privacy framework}
\label{sec:FraAndOrg}

\stitle{Privacy leakage in repeated data collection.}

Although LDP is a very strict notion of privacy, its effectiveness decreases if the data is collected repeatedly. If we collect counter values of a user $i$ for $T$ time stamps by executing an $\eps$-LDP mechanism $\cal A$ independently on each time stamp, $x_i(1)x_i(2) \ldots x_i(T)$ can be only guaranteed indistinguishable to another sequence of counter values, $x'_i(1)x'_i(2) \ldots x'_i(T)$, by a factor of up to $e^{T\cdot\eps}$, which is too large to be reasonable as $T$ increases.

Hence, in applications such as telemetry, where data is collected continuously, privacy guarantees provided by an LDP mechanism for a single round of data collection are not sufficient. We formalize our privacy guarantee to enhance LDP for repeated data collection later in \csec\ref{Sec:RepeatedCollection}. However, intuitively we ensure that every user blends with a large set of other users who have very different behaviors. Similar philosophy can be found in Blowfish privacy \cite{sigmod:HeMD14} which protects only a specified subset of pairs of neighborhood databases to trade-off privacy for utility.
On a different but relevant line of work about streaming model of DP \cite{innovations:DworkNPRY10}, the event-level private counting problem under continual observation is studied \cite{soda:Dwork10}, with almost tight upper bounds in error (polynomial in ${\log T}$) in \cite{stoc:DworkNPR10} and \cite{icalp:ChanSS10}.
\cite{pvldb:KellarisPXP14} proposes a weaker protection for continual events in $w$ consecutive timestamps to remove the dependency on the length of time period $T$.

\begin{wrapfigure}{r}{0.2\textwidth}
		\vspace{-0.6cm}
		\begin{center}
			\includegraphics[width=0.2\textwidth]{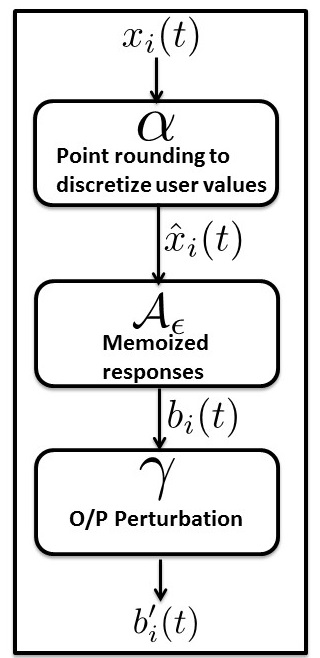}		
			\caption{Privacy Framework}
			\vspace{-1.2cm}
		\end{center}
	\end{wrapfigure}

\stitle{Our Privacy Framework and Guarantees.} Our framework for repeated private collection of counter data follows similar outline as the framework used in~\cite{ccs:ErlingssonPK14}. Our framework for mean and histogram estimation  has four main components:

1) An important building block for our overall solution are $1$-bit mechanisms that provide local $\epsilon$-LDP guarantees and good accuracy for a single round of data collection (Section~\ref{sec:mean}).

2) An $\alpha$-point rounding scheme to randomly discretize users private values prior to applying memoization (to conceal small changes) while keeping the expectation of discretized values intact (Section~\ref{Sec:RepeatedCollection}).

3) Memoization of discretized values using the $1$-bit mechanisms to avoid privacy leakage from repeated data collection (Section~\ref{Sec:RepeatedCollection}). In particular, if the counter value of a user remains approximately consistent, then the user is guaranteed $\epsilon$-differential privacy even after many rounds of data collection.

4) Finally, output perturbation (instantaneous noise in \cite{ccs:ErlingssonPK14}) to protect exposing the transition points due to large changes in user's behavior and attacks based on auxiliary information (Section~\ref{sec:opperturbation}).

We now describe these components in more detail focusing predominantly on mean estimation. Later, in Section~\ref{sec:exp} we present our experimental results and in Section~\ref{Sec:Deployment} we discuss some details of the deployment in Windows 10.


\section{Single-round LDP mechanisms for mean and histogram Estimation}\label{sec:mean}

We first describe our $1$-bit LDP mechanisms for mean and histogram estimation. Our mechanisms are inspired by the works of Duchi {\em et al.} \cite{nips:DuchiWJ13,focs:DuchiJW13,corr:DuchiWJ16} and Bassily and Smith \cite{stoc:BassilyS15}. However, our mechanisms are tuned for more efficient communication (by sending $1$ bit for each counter each time) and stronger protection in repeated data collection (introduced later in \csec\ref{Sec:RepeatedCollection}). To the best our knowledge, the exact form of mechanisms presented in this Section was not known. Our algorithms yield accuracy gains in concrete settings (see Section~\ref{sec:exp}) and are easy to understand and implement.

\subsection{1-Bit Mechanism for mean estimation}\label{sec:1bit:mean}

\stitle{Collection mechanism {\sf 1BitMean}:} When the collection of counter $x_i(t)$ at time $t$ is requested by the data collector, each user $i$ sends one bit $b_i(t)$, which is independently drawn from the distribution:
\begin{equation}\label{equ:meanmech}
b_i (t) =
\left\{
\begin{array}{ll}
1,& \mathrm{with \ probability\ } \frac{1}{e^\epsilon+1}+\frac{x_i(t)}{m}\cdot \frac{e^\epsilon -1}{e^\epsilon+1}; \\
0,& \mathrm{otherwise}.                                                           \\
\end{array}
\right.
\end{equation}

\stitle{Mean estimation.} Data collector obtains the bits $\{b_i(t)\}_{i\in [n]}$ from $n$ users and estimates $\sigma(t)$ as
\begin{equation}\label{equ:meanest}
\hat{\sigma}(t) = \frac{m}{n}\sum_{i=1}^n \frac{{b}_i(t) \cdot (e^\eps+1)-1}{e^\eps-1}.
%
\end{equation}

The basic randomizer of \cite{stoc:BassilyS15} is equivalent to our 1-bit mechanism for the case when each user takes values either $0$ or $m$. The above mechanism can also be seen as a simplification of the multidimensional mean-estimation mechanism given in \cite{focs:DuchiJW13}. For the 1-dimensional mean estimation, Duchi {\em et al.} \cite{focs:DuchiJW13} show that Laplace mechanism is {\em asymptotically optimal} for the mini-max error. However, the communication cost per user in Laplace mechanism is $\Omega(\log m)$ bits, and our experiments show it also leads to larger error compared to our 1-bit mechanism.
We prove following results for the above 1-bit mechanism.
\begin{theorem}\label{thm:singlemean}
	For single-round data collection, the mechanism {\sf 1BitMean} in \eqref{equ:meanmech} preserves $\epsilon$-LDP for each user. Upon receiving the $n$ bits $\{b_i(t)\}_{i\in [n]}$, the data collector can then estimate the mean of counters from $n$ users as $\hat{\sigma}(t)$ in \eqref{equ:meanest}. With probability at least $1-\delta$, we have \[|\hat{\sigma}(t) - \sigma(t)| \leq \frac{m}{\sqrt{2n}} \cdot \frac{e^\eps+1}{e^\eps-1} \cdot \sqrt{\log\frac{2}{\delta}}.\]
\end{theorem}
We establish a few lemmas first and then prove Theorem~\ref{thm:singlemean}.
\begin{lemma}\label{lem:Privacy}
	The algorithm {\sf 1BitMean} preserves $\epsilon$-DP of every user.
\end{lemma}
\begin{proof}
	Observe that each user contributes only a single bit $b_i$ to data collector. By formula~\eqref{equ:meanmech} the probability that $b_i=0$ varies from $\frac{1}{e^\epsilon+1}$ to $\frac{e^\epsilon}{e^\epsilon+1}$ depending on the private value $x_i.$ Similarly, the probability that $b_i=1$ varies from $\frac{1}{e^\epsilon+1}$ to $\frac{e^\epsilon}{e^\epsilon+1}$ with $x_i.$ Thus the ratios of respective probabilities for different values of $x_i$ can be at most $e^\epsilon.$
\end{proof}

Recall the definition of $\sigma$.
\begin{equation}
\label{Eqn:Sigma}
\sigma = \frac{1}{n} \sum_{i}x_i
\end{equation}

\begin{lemma}\label{Lemma:Unbiased}
	$\hat \sigma (t)$ in Equation (\ref{equ:meanest}) is an unbiased estimator for $\sigma.$
\end{lemma}
\begin{proof}
	Observe that
	\begin{eqnarray}
	\mathbf{E}[\hat \sigma (t)] & = & \frac{m}{n}\cdot \frac{e^\epsilon+1}{e^\epsilon-1}\cdot \left(\sum_{i\in [n]}\mathbf{E}[b_i(t)]-\frac{n}{e^\epsilon+1}\right) \nonumber \\
	& = &  \frac{m}{n}\cdot \frac{e^\epsilon+1}{e^\epsilon-1}\cdot \left( \sum_{i\in [n]}\frac{x_i(t)}{m}\cdot \frac{e^\epsilon -1}{e^\epsilon+1} \right) \nonumber \\
	& = &  \sigma(t). \nonumber
	\end{eqnarray}
\end{proof}

\begin{lemma}\label{lem:Accuracy}
	Let $\hat \sigma(t)$ and $\sigma$ be as in~Equations (\ref{equ:meanest}, \ref{Eqn:Sigma}). Let $\theta\in (0,1)$ be arbitrary. We have
	\begin{equation}\label{Eqn:Main}
	\mathrm{Pr}\left[\left|\hat \sigma(t) - \sigma(t)\right|\geq \theta m \right]\leq 2\cdot e^{-2\theta^2\cdot n\cdot\left(\frac{e^\epsilon-1}{e^\epsilon+1}\right)^2}.
	\end{equation}
\end{lemma}
\begin{proof}
	Clearly, for every $i\in [n],$ we have
	$$\mathbf{E}[b_i(t)] = \frac{1}{e^\epsilon+1}+\frac{x_i(t)}{m}\cdot \frac{e^\epsilon -1}{e^\epsilon+1}.$$
	Let
	\begin{equation}\label{Eqn:Mu}
	\mu = \mathbf{E}\left(\sum_{i\in [n]} b_i(t)\right) = \frac{n}{e^\epsilon+1} + \frac{n\sigma(t)}{m}\cdot \frac{e^\epsilon -1}{e^\epsilon+1}.
	\end{equation}
	Applying the Chernoff-Hoeffding bound~\cite[Theorem 2.8]{boucheron2013concentration} to independent $\{0,1\}$-random variables $\{b_i(t)\}_{i\in n},$ for all $t>0,$ we have
	\begin{equation}\label{Eqn:Chernoff_0}
	\mathrm{Pr}\left[\left|\sum_{i\in [n]}b_i(t) - \mu\right|\geq t\right]\leq 2\cdot e^{-\frac{2t^2}{n}}.
	\end{equation}
	Combining~(\ref{Eqn:Chernoff_0}) and~(\ref{Eqn:Mu}) we get
	\begin{equation}\label{Eqn:Chernoff_1}
	\mathrm{Pr}\left[\left|\sum_{i\in [n]}b_i(t) - \frac{n}{e^\epsilon+1} - \frac{n\sigma}{m}\cdot \frac{e^\epsilon -1}{e^\epsilon+1} \right|\geq t\right]\leq  e^{-\frac{2t^2}{n}}.
	\end{equation}
	Combining~(\ref{Eqn:Chernoff_1}),~(\ref{Eqn:Sigma}), and~(\ref{equ:meanest}) we conclude
	\begin{equation}\label{Eqn:Chernoff_2}
	\mathrm{Pr}\left[\left|\hat \sigma(t) - \sigma(t)\right|\geq t\cdot \frac{m}{n}\cdot \frac{e^\epsilon +1}{e^\epsilon -1}\right]\leq 2\cdot  e^{-\frac{2t^2}{n}}.
	\end{equation}
	Thus setting $t=\theta n \cdot \frac{e^\epsilon -1}{e^\epsilon +1}$ we obtain
	\begin{equation}\label{Eqn:Chernoff_3}
	\mathrm{Pr}\left[\left|\hat \sigma(t) - \sigma(t)\right|\geq \theta m \right]\leq 2\cdot e^{-2\theta^2\cdot n\cdot\left(\frac{e^\epsilon-1}{e^\epsilon+1}\right)^2},
	\end{equation}
	which concludes the proof.
\end{proof}

\begin{proof}[Proof of Theorem \ref{thm:singlemean}]
For any $\delta \in [0,1]$, set $\delta =  2\cdot e^{-2\theta^2\cdot n\cdot\left(\frac{e^\epsilon-1}{e^\epsilon+1}\right)^2}$. Then, error 
$$\theta m \leq \frac{m}{\sqrt{2n}} \cdot \frac{e^\eps+1}{e^\eps-1} \cdot \sqrt{\log\frac{2}{\delta}}.$$ 
This fact combined with Lemmas (\ref{lem:Privacy}, \ref{lem:Accuracy}) completes the proof.
\end{proof}

\subsection{$d$-Bit Mechanism for histogram estimation}
\label{sec:1bit:histogram}

Now we consider the problem of estimating histograms of counter values in a discretized domain with $k$ buckets with LDP to be guaranteed.

This problem has extensive literature both in computer science and statistics, and dates back to the seminal work Warner \cite{warner1965randomized}; we refer the readers to following excellent papers \cite{icalp:HsuKR12, nips:DuchiWJ13, stoc:BassilyS15, kairouz2016discrete} for more information.
Recently, Bassily and Smith \cite{stoc:BassilyS15} gave asymptotically tight results for the problem in the worst-case model building on the works of \cite{icalp:HsuKR12}. On the other hand, Duchi {\em et al.} \cite{nips:DuchiWJ13} introduce a mechanism by adapting Warner's classical randomized response mechanism in \cite{warner1965randomized}, which is shown to be optimal for the statistical mini-max regret if one does not care about the cost of communication.

The generic mechanism introduced in \cite{stoc:BassilyS15} can be used to reduce the communication cost in Duchi {\em et al.}'s mechanism to 1 bit per user, which however only works for $\eps \leq 2$.\footnote{\cite{stoc:BassilyS15} requires $\eps \leq \ln 2$ but we can optimize the parameters to loose the constraint to $\eps \leq 2$.} Another major technical component in \cite{stoc:BassilyS15} is the use of {\em Johnson-Lindenstrauss} lemma to make the communication cost polynomial in $\log k$. This component seems very difficult to be used in practice, because it requires from each user $\bigoh{nk}$ storage per counter, and/or $\bigoh{nk}$ time per collection. In our applications, $n$ (the number of users) is order of millions, and thus makes their mechanism prohibitively expensive. \cite{kairouz2016discrete} generalizes Warner's randomized response mechanism from binary to $k$-ary, which is close to optimal for large $\eps$ but sub-optimal for small $\eps$.

Therefore, in order to have a smooth trade-off between accuracy and communication cost (as well as the ability to protect privacy in repeated data collection, which will be introduced in \csec\ref{Sec:RepeatedCollection}) we introduce a modified version of Duchi {\em et al.}'s mechanism \cite{nips:DuchiWJ13} based on subsampling by buckets.

\stitle{Collection mechanism {\sf $d$BitFlip}:} Each user $i$ randomly draws $d$ bucket numbers without replacement from $[k]$, denoted by $j_1, j_2, \ldots, j_d$. When the collection of discretized bucket number $v_i(t) \in [k]$ at time $t$ is requested by the data collector, each user $i$ sends a vector:
\begin{align*}
& b_i(t) = \left[(j_1, b_{i, j_1}(t)), (j_2, b_{i, j_2}(t)), \ldots, (j_d, b_{i, j_d}(t))\right], ~\text{where $b_{i, j_p}(t)$ is a random 0-1 bit,}
\\
& \text{with}~ \pr{b_{i, j_p}(t) = 1} =
\begin{cases}
{e^{\eps/2}}/{(e^{\eps/2}+1)} & \text{if}~ v_i(t) = j_p
\\
{1}/{(e^{\eps/2}+1)} & \text{if}~ v_i(t) \neq j_p
\end{cases}, \text{~for $p=1,2,\ldots,d$}.
\end{align*}
Under the same public coin model as in \cite{stoc:BassilyS15}, each user $i$ only needs to send to the data collector $d$ bits $b_{i, j_1}(t)$, $b_{i, j_2}(t)$, $\ldots$, $b_{i, j_d}(t)$ in $b_i(t)$, as $j_1, j_2, \ldots, j_d$ can be generated using public coins.

\stitle{Histogram estimation.} Data collector estimates histogram $h_t$ as: for $v \in [k],$
\begin{equation}\label{equ:histogramest}
\hat{h}_t(v) = \frac{k}{nd} \sum_{{b}_{i,v}(t)~\text{is received}} \!\!\!\!\!\!\!\!\!\! \frac{{b}_{i,v}(t) \cdot (e^{\eps/2}+1)-1}{e^{\eps/2}-1}.
\end{equation}

When $d = k$, {\sf $d$BitFlip} is exactly the same as the one in in Duchi {\em et al.}\cite{nips:DuchiWJ13}. The privacy guarantee is straightforward. In terms of the accuracy, the intuition is that for each bucket $v \in [k]$, there are roughly $nd/k$ users responding with a 0-1 bit $b_{i,v}(t)$. We can prove the following result.
\begin{theorem}\label{thm:singlehistogram}
	For single-round data collection, the mechanism {\sf $d$BitFlip} preserves $\epsilon$-LDP for each user. Upon receiving the $d$ bits $\{b_{i,j_p}(t)\}_{p\in[d]}$ from each user $i$, the data collector can then estimate then histogram $h_t$ as $\hat h_t$ in \eqref{equ:histogramest}. With probability at least $1-\delta$, we have,
	\[
	\max_{v \in [k]}|h_t(v) - \hat h_t(v)| \leq \sqrt{\frac{5k}{nd}} \cdot \frac{e^{\eps/2}+1}{e^{\eps/2}-1} \cdot \sqrt{\log\frac{6k}{\delta}} \leq \bigoh{\sqrt{\frac{k\log(k/\delta)}{\eps^2 nd}}}.
	\]
\end{theorem}
\begin{proof}
The privacy guarantee of our algorithm is straightforward from the construction. To analyze the error bound $|h_t(v) - \hat h_t(v)|$ for each $v \in [k]$, let us consider the set of users $U(v)$ each of whom sends $(v, {b}_{i,v}(t))$ to the data collector. Let $n_v = |U(v)|$ and based on how each user chooses $j_1, \ldots, j_d$, we know $n_v=\frac{k}{nd}$ in expectation. Consider $h'_t(v) = \frac{1}{n_v} \cdot |\{i : v_i(t) = v ~\text{and}~ i \in U(v)\}|$; since $U(v)$ can be considered as a uniform random sample from $[n]$,  we can show using the Hoeffding’s inequality that
\[
|h_t(v) - h'_t(v)| \leq \bigoh{\sqrt{\frac{k\log(1/\delta)}{nd}}} ~\text{with probability at least}~ 1-\delta/2.
\]
From \eqref{equ:histogramest} and, again from the Hoeffding’s inequality, we have
\[
|h'_t(v) - \hat{h}_t(v)| \leq \bigoh{\sqrt{\frac{k\log(1/\delta)}{\eps^2 nd}}} ~\text{with probability at least}~ 1-\delta/2.
\]
Putting them together, and using the union bound and the triangle inequality, we have
\[
|h_t(v) - \hat{h}_t(v)| \leq \bigoh{\sqrt{\frac{k\log(1/\delta)}{\eps^2 nd}}} ~\text{with probability at least}~ 1-\delta.
\]
The bound of $\max_{v \in [k]}|h_t(v) - \hat{h}_t(v)|$ follows from the union bound over the $k$ buckets.
\end{proof}



\section{Memoization for continual collection of counter data}\label{Sec:RepeatedCollection}

One important concern regarding the use of $\epsilon$-LDP algorithms (e.g., in Section~\ref{sec:1bit:mean}) to collect counter data pertains to privacy leakage that may occur if we collect user's data repeatedly (say, daily) and user's private value $x_i$ does not change or changes little. Depending on the value of $\epsilon,$ after a number of rounds, data collector will have enough noisy reads to estimate $x_i$ with high accuracy.

Memoization~\cite{ccs:ErlingssonPK14} is a simple rule that says that: {\it At the account setup phase each user pre-computes and stores his responses to data collector for all possible values of the private counter. At data collection users do not use fresh randomness, but respond with pre-computed responses corresponding to their current counter values.} Memoization (to a certain degree) takes care of situations when the private value $x_i$ stays constant. Note that the use of memoization violates differential privacy. If memoization is employed, data collector can easily distinguish a user whose value keeps changing, from a user whose value is constant; no matter how small the $\epsilon$ is. However, privacy leakage is limited. When data collector observes that user's response had changed, this only indicates that user's value had changed, but not what it was and not what it is.

As observed in~\cite[Section 1.3]{ccs:ErlingssonPK14} using memoization technique in the context of collecting counter data is problematic for the following reason. Often, from day to day, private values $x_i$ do not stay constant, but rather experience small changes (e.g., one can think of app usage statistics reported in seconds). Note that, naively using memoization adds no additional protection to the user whose private value varies but stays approximately the same, as data collector would observe many independent responses corresponding to it.

One naive way to fix the issue above is to use discretization: pick a large integer (segment size) $s$ that divides $m;$ consider the partition of all integers into segments $[\ell s,(\ell+1)s];$ and have each user report his value after rounding the true value $x_i$ to the mid-point of the segment that $x_i$ belongs to. This approach takes care of the issue of leakage caused by small changes to $x_i$ as users values would now tend to stay within a single segment, and thus trigger the same memoized response; however accuracy loss may be extremely large. For instance, in a population where all $x_i$ are $\ell s+1$ for some $\ell,$ after rounding every user would be responding based on the value $\ell s+s/2.$

In the following subsection we present a better (randomized) rounding technique (termed $\alpha$-point rounding) that has been previously used in approximation algorithms literature~\cite{goemans2002single, bansal2008improved} and rigorously addresses the issues discussed above. We first consider the mean estimation problem.

\subsection{$\alpha$-point rounding for mean estimation}\label{SubSec:AlphaRound}

The key idea of rounding is to discretize the domain where users' counters take their values. Discretization reduces domain size, and users that behave consistently take less different values, which allows us to apply memoization to get a strong privacy guarantee.

As we demonstrated above discretization may be particularly detrimental to accuracy when users' private values are correlated. We propose addressing this issue by: {\it making the discretization rule independent across different users}. This ensures that when (say) all users have the same value, some users round it up and some round it down, facilitating a smaller accuracy loss.

We are now ready to specify the algorithm that extends the basic algorithm {\sf 1BitMean} and employs both $\alpha$-point rounding and memoization. We assume that counter values range in $[0,m].$
\begin{enumerate}
	\item At the algorithm design phase, we specify an integer $s$ (our discretization granularity). We assume that $s$ divides $m.$ We suggest setting $s$ rather large compared to $m,$ say $s=m/20$ or even $s=m$ depending on the particular application domain.
	
	\item At the the setup phase, each user $i\in [n]$ independently at random picks a value $\alpha_i\in\{0,\ldots,s-1\},$
	that is used to specify the rounding rule.
	
	\item User $i$ invokes the basic algorithm {\sf 1BitMean} with range $m$ to compute and memoize $1$-bit responses to data collector for all $\frac{m}{s}+1$ values $x_i$ in the arithmetic progression
	\begin{equation}\label{Eqn:Ai}
	A=\{\ell s\}_{0 \leq \ell\leq \frac{m}{s}}.
	\end{equation}
	
	\item Consider a  user $i$ with private value $x_i$ who receives a data collection request. Let $x_i \in [L, R)$, where $L, R$ are the two neighboring elements of the arithmetic progression $\{\ell s\}_{0 \leq \ell\leq \frac{m}{s}+1}.$ The user $x_i$ rounds value to $L$ if $x_i + \alpha_i < R$; otherwise, the user rounds the value to $R$. Let $y_i$ denote the value of the user after rounding. In each round, user responds with the memoized bit for value $y_i$. Note that rounding is always uniquely defined.	
\end{enumerate}

We now establish the properties of the algorithm above.

\begin{lemma}
	\label{Lem:RRounding}
	Define $\sigma\prime := \frac{1}{n} \sum_i {y_i}$.  Then, ${\mathbb E} [\sigma\prime] = \sigma,$ where $\sigma$ is defined by~(\ref{Eqn:Sigma}).
\end{lemma}
\begin{proof}
	Let $a = x_i - L$ and $b = R - x_i$. Define a random variable $z_i$ as follows. Let $z_i = b$ with probability $a/(a+b)$ and $z_i = - a$ with probability $b/(a+b).$
	Then, ${\mathbb E} [z_i] = 0$. It is easy to verify that random variable  $y_i$ can be rewritten as  $y_i := x_i + z_i$. The proof the lemma follows from the linearity of expectation and the fact that  ${\mathbb E} [z_i] = 0$.
\end{proof}

Perhaps a bit surprisingly, using $\alpha$-point rounding does not lead to additional accuracy losses independent of the choice of discretization granularity $s.$	

\begin{theorem}\label{Th:RRandPM}
	Independent of the value of discretization granularity $s,$ at any round of data collection, the algorithm above provides the same accuracy guarantees as given in Theorem~\ref{thm:singlemean}.
\end{theorem}
\begin{proof}
	It suffices to show that independent of the $s,$ each output bit $b_i$ is still sampled according to the distribution given by formula~(\ref{equ:meanmech}). We use the notation of Lemma~\ref{Lem:RRounding}. Be formula~(\ref{equ:meanmech}) and the definition of $b_i$ we have:
	\begin{eqnarray}\label{Eqn:SisIrrelevenat}
	\mathrm{Pr}[b_i=1] & = &
	\frac{b}{a+b}\left(\frac{1}{e^\epsilon+1}+\frac{L}{m}\cdot \frac{e^\epsilon-1}{e^\epsilon+1}\right)+
	\frac{b}{a+b}\left(\frac{1}{e^\epsilon+1}+\frac{R}{m}\cdot \frac{e^\epsilon-1}{e^\epsilon+1}\right) \nonumber \\
	& = &
	\frac{1}{e^\epsilon+1} + \left(\frac{b}{a+b}\cdot \frac{L}{m}+\frac{a}{a+b}\cdot\frac{R}{m}\right)\cdot \left(\frac{e^\epsilon-1}{e^\epsilon+1}\right) \nonumber \\
	& = &
	\frac{1}{e^\epsilon+1} + \frac{1}{m}\cdot\left(\frac{b(x_i-a)+a(x_i+b)}{a+b}\right)\cdot \left(\frac{e^\epsilon-1}{e^\epsilon+1}\right) \nonumber \\
	& = &
	\frac{1}{e^\epsilon+1} + \frac{x_i}{m}\cdot \frac{e^\epsilon-1}{e^\epsilon+1}, \nonumber
	\end{eqnarray}
	which concludes the proof.
\end{proof}

\subsection{Privacy definition using permanent memoization}\label{SubSec:PrivacyGuarantees}
In what follows we detail privacy guarantees provided by an algorithm that employs $\alpha$-point rounding and memoization in conjunction with the $\epsilon$-DP $1$-bit mechanism of Section~\ref{sec:1bit:mean} against a data collector that receives a very long stream of user's responses to data collection events.

Let $U$ be a user and $x(1),\ldots, x(T)$ be the sequence of $U$'s private counter values. Given user's private value $\alpha_i,$ each of $\{x(j)\}_{j\in [T]}$ gets rounded to the corresponding value $\{y(j)\}_{j\in [T]}$ in the set $A$ (defined by~(\ref{Eqn:Ai})) according to the rule given in Section~\ref{SubSec:AlphaRound}.

\begin{definition}\label{Def:BehPatt}
	Let $B$ be the space of all sequences $\{z(j)\}_{j\in [T]} \in A^T,$ considered up to an arbitrary permutation of the elements of $A.$ We define the behavior pattern $b(U)$ of the user $U$ to be the element of $B$ corresponding to $\{y(j)\}_{j\in [T]}.$ We refer to the number of distinct elements $y(j)$ in the sequence $\{y(j)\}_{j\in [T]}$ as the width of $b(U).$
\end{definition}

We now discuss our notion of behavior pattern, using counters that carry daily app usage statistics as an example. Intuitively, users map to the same behavior pattern if they have the same number of different modes (approximate counter values) of using the app, and switch between these modes on the same days. For instance, one user that uses an app for 30 minutes on weekdays, 2 hours on weekends, and 6 hours on holidays, and the other user who uses the app for 4 hours on weekdays, 10 minutes on weekends, and does not use it on holidays will likely map to the same behavior pattern. Observe however that the mapping from actual private counter values $\{x(j)\}$ to behavior patterns is randomized, thus there is a likelihood that some users with identical private usage profiles may map to different behavior patterns. This is a positive feature of the Definition~\ref{Def:BehPatt} that increases entropy among users with the same behavior pattern.

The next theorem shows that the algorithm of Section~\ref{SubSec:AlphaRound} makes users with the same behavior pattern blend with each other from the viewpoint of data collector (in the sense of differential privacy).

\begin{theorem}\label{Th:PermMem}
	Consider users $U$ and $V$ with sequences of private counter values $\{x_U(1),\ldots, x_U(T)\}$ and $\{x_V(1),\ldots, x_V(T)\}.$  Assume that both $U$ and $V$ respond to $T$ data collection events using the algorithm presented in Section~\ref{SubSec:AlphaRound}, and $b(U)=b(V)$ with the width of $b(U)$ equal to $w$. Let ${\bf s}_U,{\bf s}_V\in \{0,1\}^T$ be the random sequences of responses generated by users $U$ and $V;$ then for any binary string ${\bf s}\in \{0,1\}^T$ in the response domain, we have:	
	\begin{equation}\label{Eqn:PermMem}
	\pr{{\bf s}_U={\bf s}}\leq e^{w\epsilon}\cdot \pr{{\bf s}_V={\bf s}}.
	\end{equation}
\end{theorem}

	\begin{proof}
		Let $\{y_U(1),\ldots, y_U(T)\}$ and $\{y_V(1),\ldots, y_V(T)\}$ be the sequences of $U's$ and $V's$ counter values after applying $\alpha$-point rounding. Since the width of $b(U)$ is $w,$ the set $\{y_U(j)\}_{j\in [T]}$ contains $w$ elements $\{y_U(j_1),\ldots,y_U(j_w)\}.$ Similarly, the set $\{y_V(j)\}_{j\in [T]}$ contains $w$ elements $\{y_V(j_1),\ldots,y_V(j_w)\}.$ Note that vectors ${\bf s}_U$ and ${\bf s}_V$ are each determined by $w$ bits that are $U$'s ($V$'s) memoized responses corresponding to counter values $\{y_U(j_s)\}_{s\in [w]}$ and $\{y_V(j_s)\}_{s\in [w]}$. By the $\epsilon$-LDP property of the basic algorithm {\sf 1BitMean} of Section~\ref{sec:1bit:mean} for all values of $y,y^\prime \in [0,\ldots,m]$ and all $b\in \{0,1\},$ we have
		$$
		\mathrm{Pr}[{\sf 1BitMean}(y)=b]\leq e^{\epsilon}\cdot \mathrm{Pr}[{\sf 1BitMean}(y^\prime)=b].
		$$
		Thus the probability of observing some specific $w$ responses of ${\mathcal A}$ can increase by at most $e^{w\epsilon}$ as we vary the inputs.
	\end{proof}

\subsubsection{Setting parameters}\label{SubSubSec:SettingS}

The $\epsilon$-LDP guarantee provided by Theorem~\ref{Th:PermMem} ensures that each user is indistinguishable from other users with the same behavior pattern (in the sense of LDP). The exact shape of behavior patterns is governed by the choice of the parameter $s.$ Setting $s$ very large, say $s=m$ or $s=m/2$ reduces the number of possible behavior patterns and thus increases the number of users that blend by mapping to a particular behavior pattern. It also yields stronger guarantee for blending within a pattern since for all users $U$ we necessarily have $b(U)\leq m/s+1$ and thus by Theorem~\ref{Th:PermMem} the likelihood of distinguishing users within a pattern is trivially at most $e^{(m/s+1)\cdot\epsilon}.$ At the same time there are cases where one can justify using smaller values of $s.$ In fact, consistent users, i.e., users whose private counter always land in the vicinity of one of a small number of fixed values enjoy a strong LDP guarantee within their patterns irrespective of $s$ (provided it is not too small), and smaller $s$ may be advantageous to avoid certain attacks based on auxiliary information as the set of all possible values of a private counter $x_i$ that lead to a specific output bit $b$ is potentially more complex.

Finally, it is important to stress that the $\epsilon$-LDP guarantee established in Theorem~\ref{Th:PermMem} is not a panacea, and in particular it is a weaker guarantee (provided in a much more challenging setting) than just the $\epsilon$-LDP guarantee across all users that we provide for a single round of data collection. In particular, while LDP across all population of users is resilient to any attack based on auxiliary information, LDP across a sub population may be vulnerable to such attacks and additional levels of protection may need to be applied. In particular, if data collector observes that user's response has changed; data collector knows with certainty that user's true counter value had changed. In the case of app usage telemetry this implies that app has been used on one of the days. This attack is partially mitigated by the output perturbation technique that is discussed in Section~\ref{sec:opperturbation}.

\subsubsection{Experimental study}\label{SubSubSec:PatternDistribution}

We use a real-world dataset of 3 million users with their daily usage of two apps (App A and B) collected (in seconds) over a continuous period of 31 days to demonstrate the mapping of users to behavior patterns in~\cfig\ref{fig:patterns}. For each behavior pattern (\cdef\ref{Def:BehPatt}), we calculate its {\em support} as {\em the number of users with their sequences in this pattern} ($y$-axis). All the patterns' supports ${sup}$ are plotted in the decreasing order, and we can also calculate the percentage of users ($x$-axis) in patterns with supports at least ${sup}.$ We vary the parameter $s$ in permanent memoization from $m$ (maximizing blending) to $m/3$ and report the corresponding distributions of pattern supports in~\cfig\ref{fig:patterns}.

\begin{figure}[t]
	\center
	\subfigure{
		\includegraphics[width=0.31\textwidth]{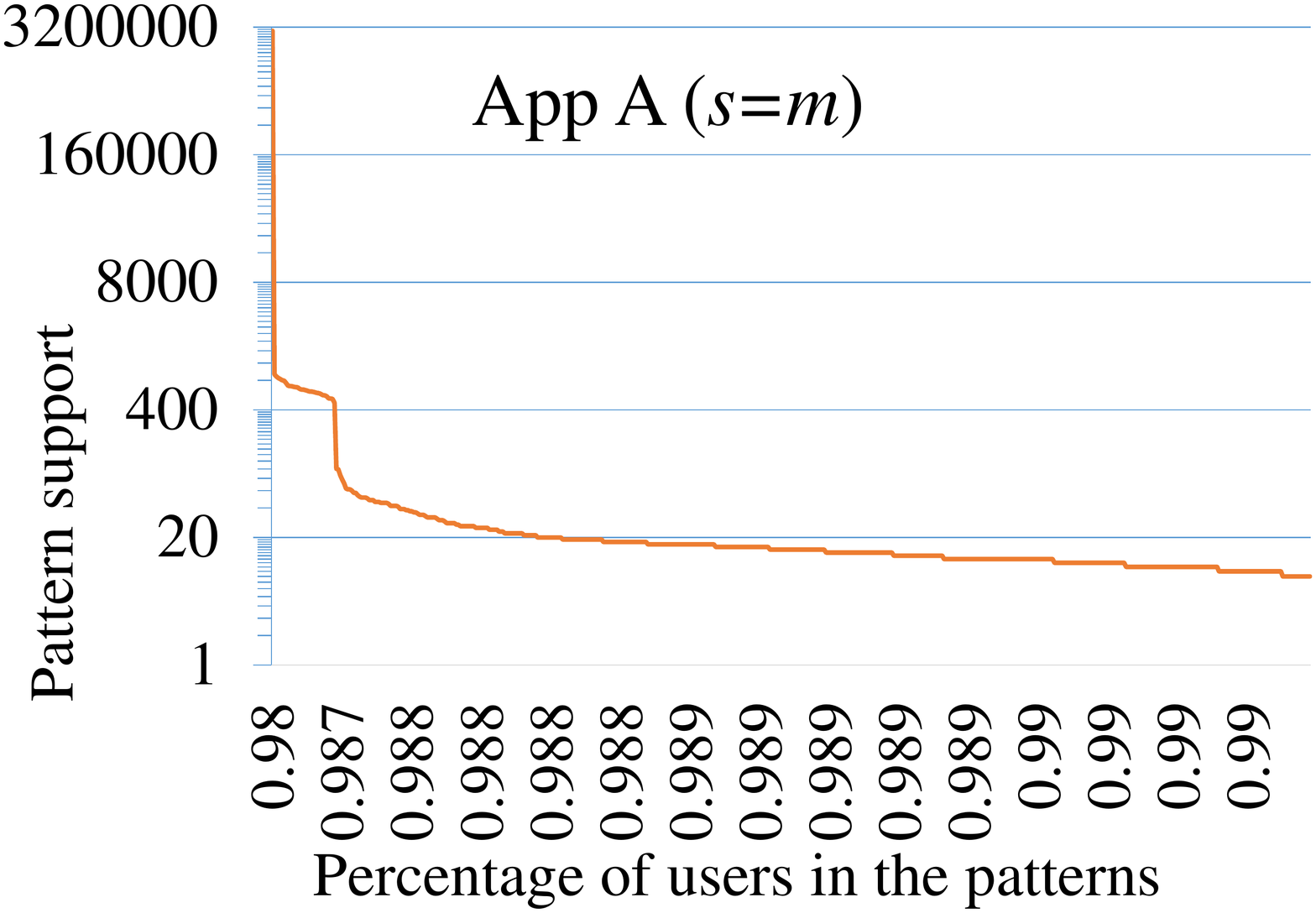}
	}
	\subfigure{
		\includegraphics[width=0.31\textwidth]{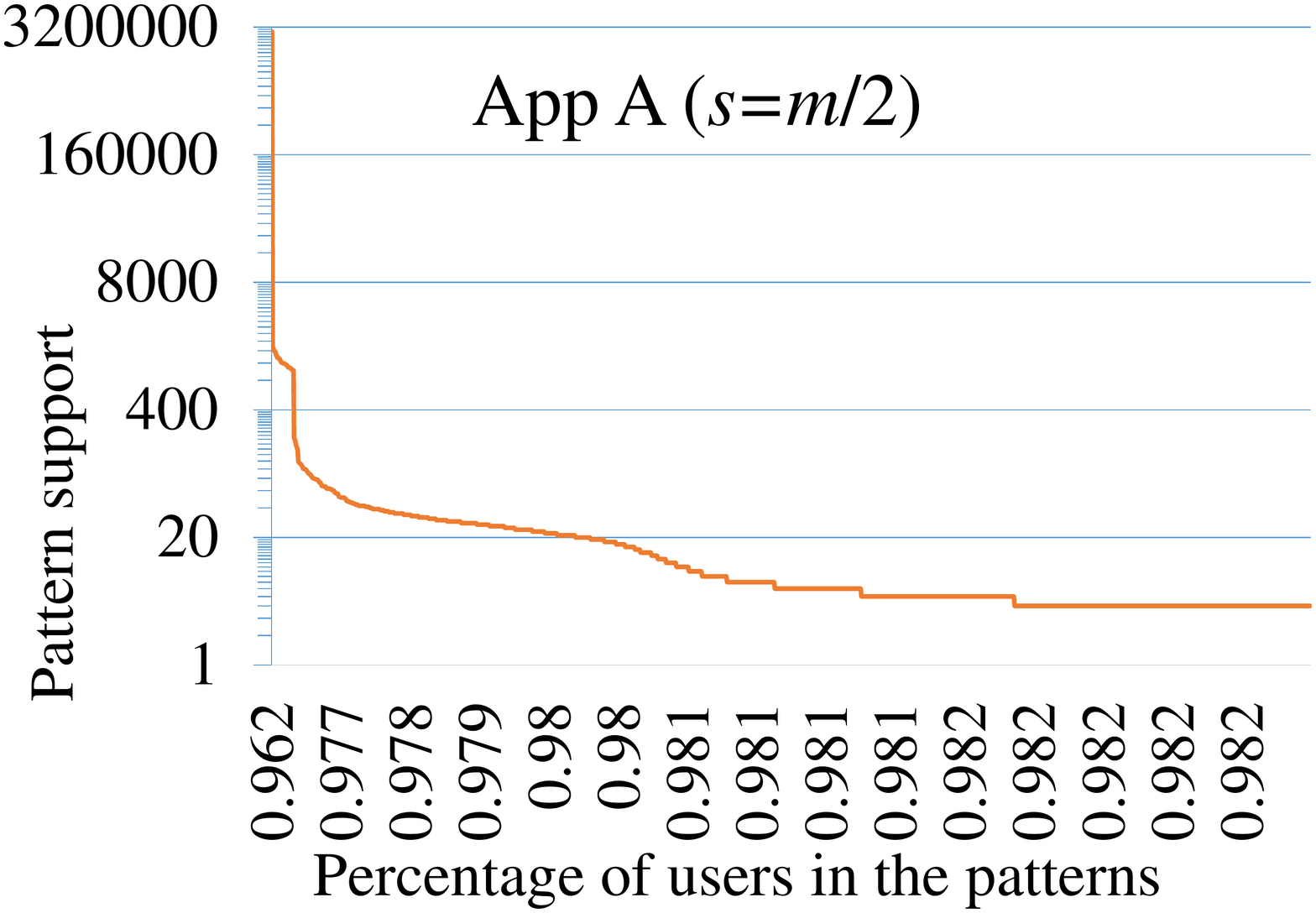}
	}
	\subfigure{
		\includegraphics[width=0.31\textwidth]{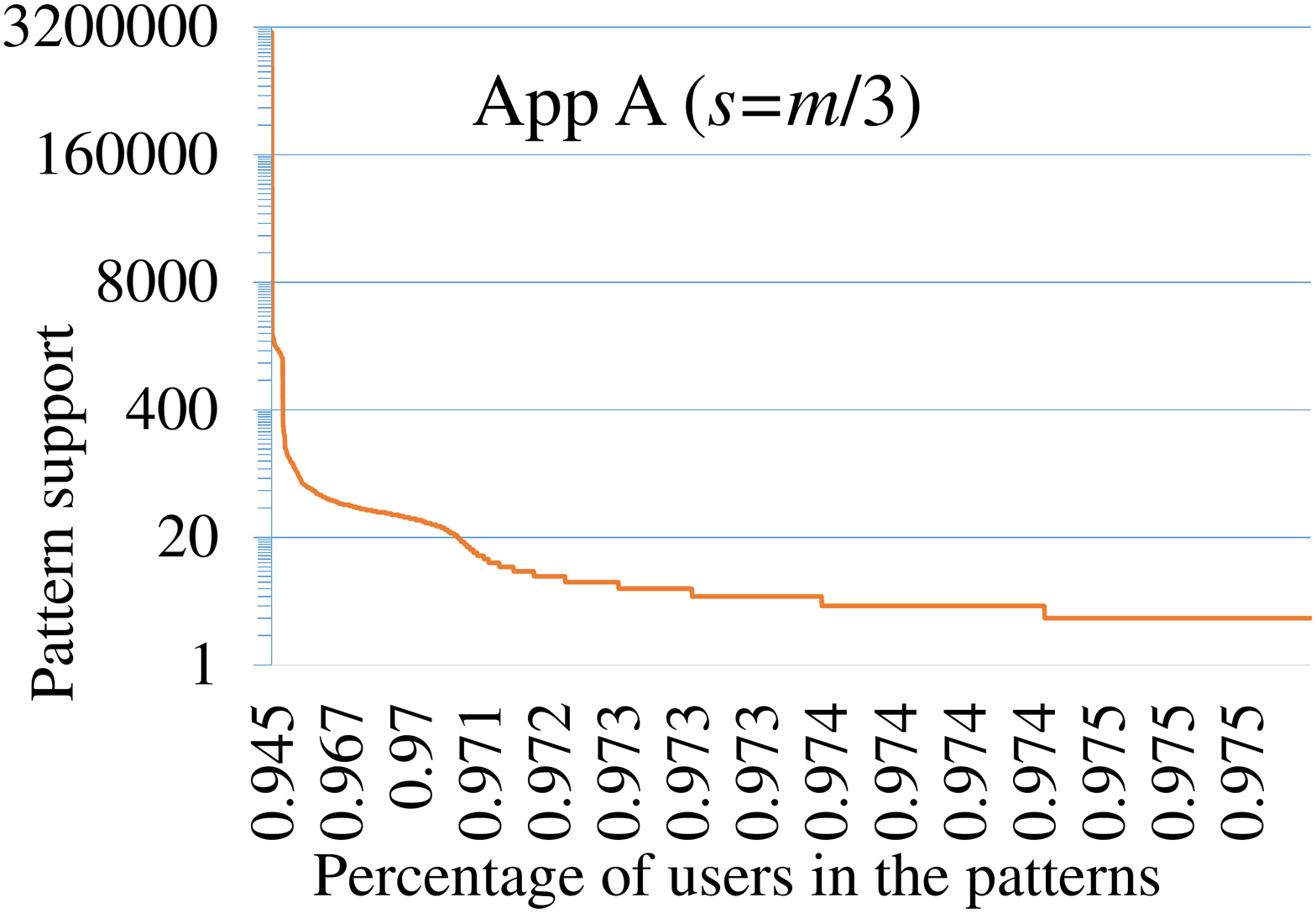}
	}
	
	\subfigure{
	\includegraphics[width=0.31\textwidth]{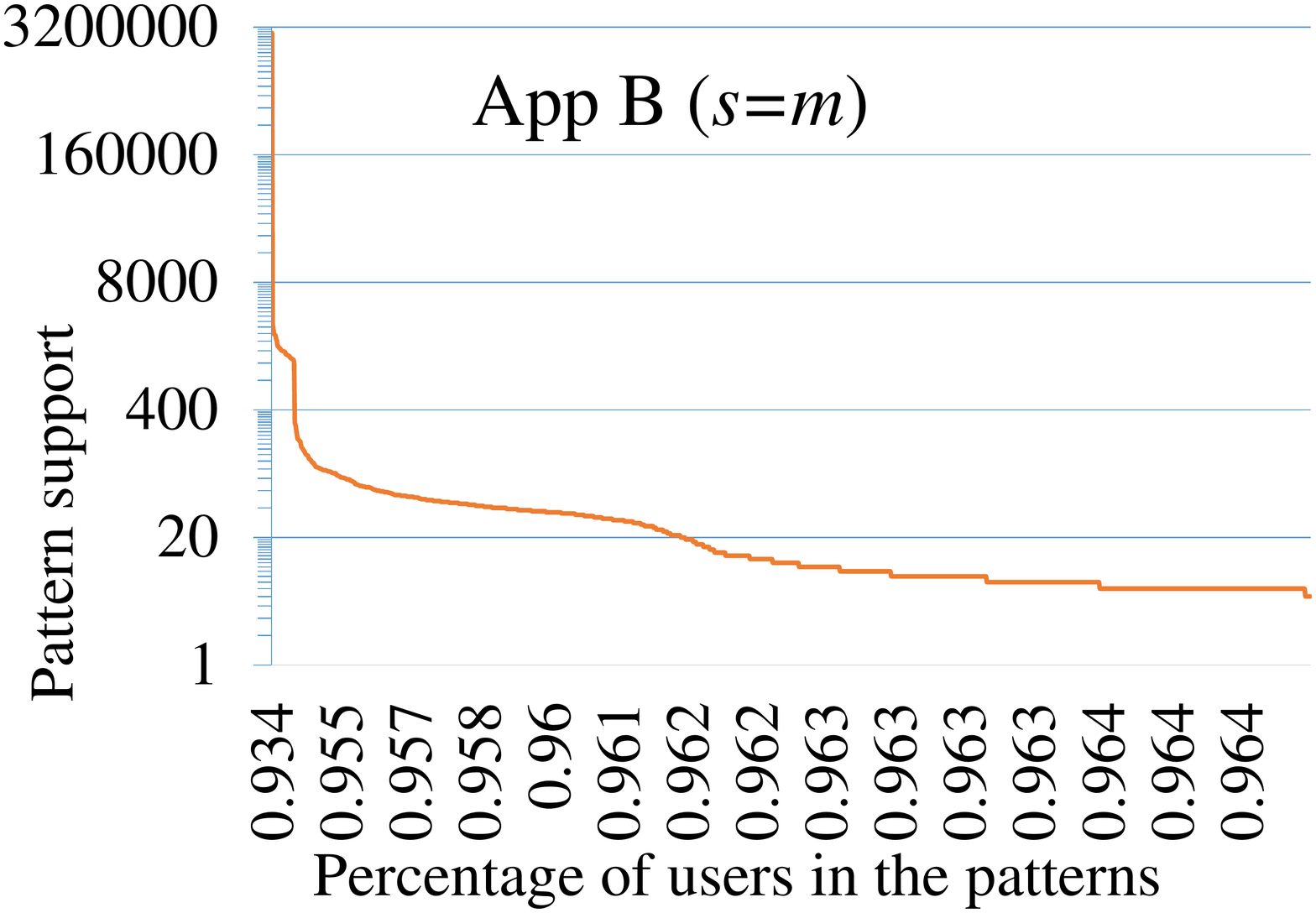}
	}
	\subfigure{
	\includegraphics[width=0.31\textwidth]{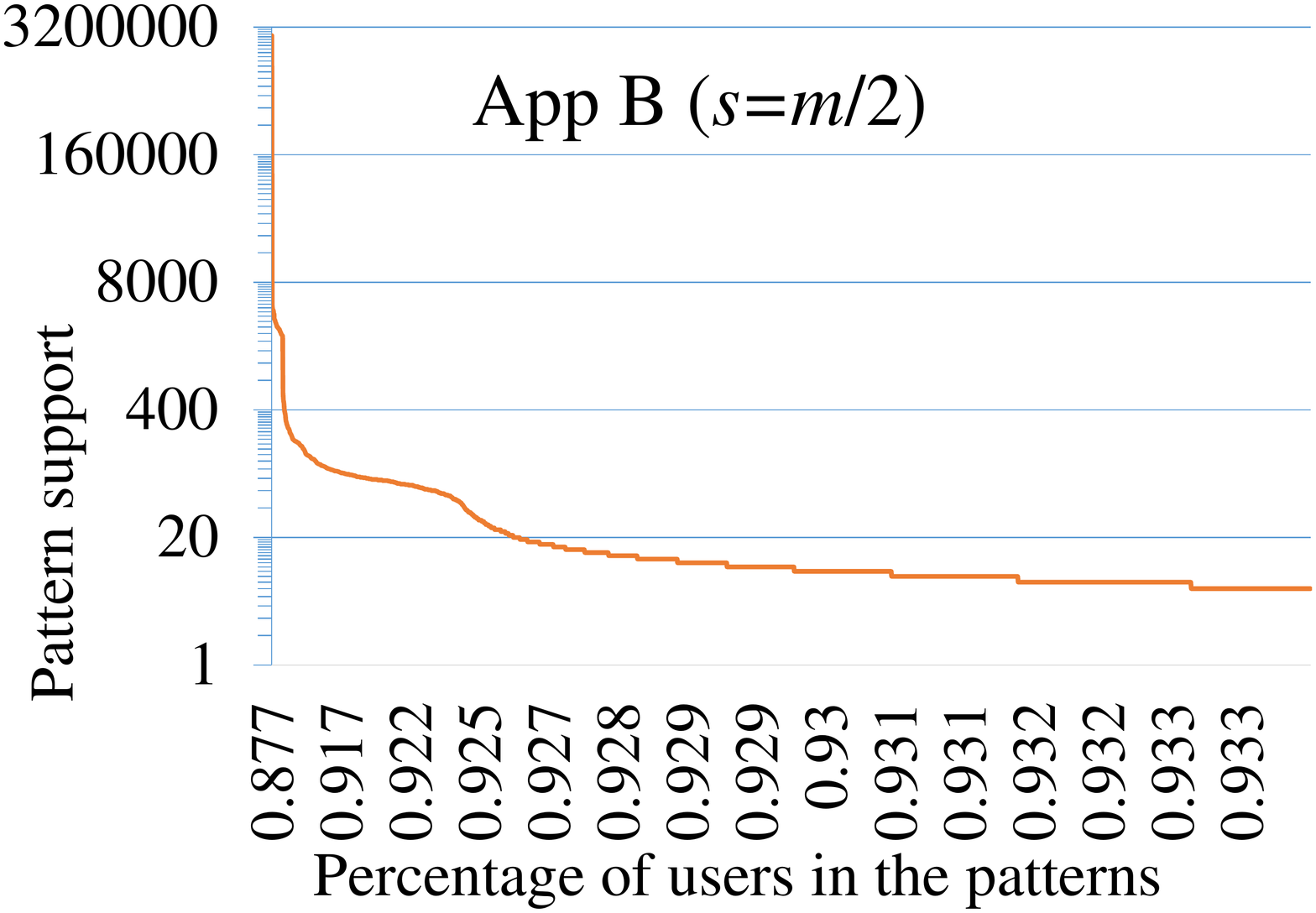}
	}
	\subfigure{
	\includegraphics[width=0.31\textwidth]{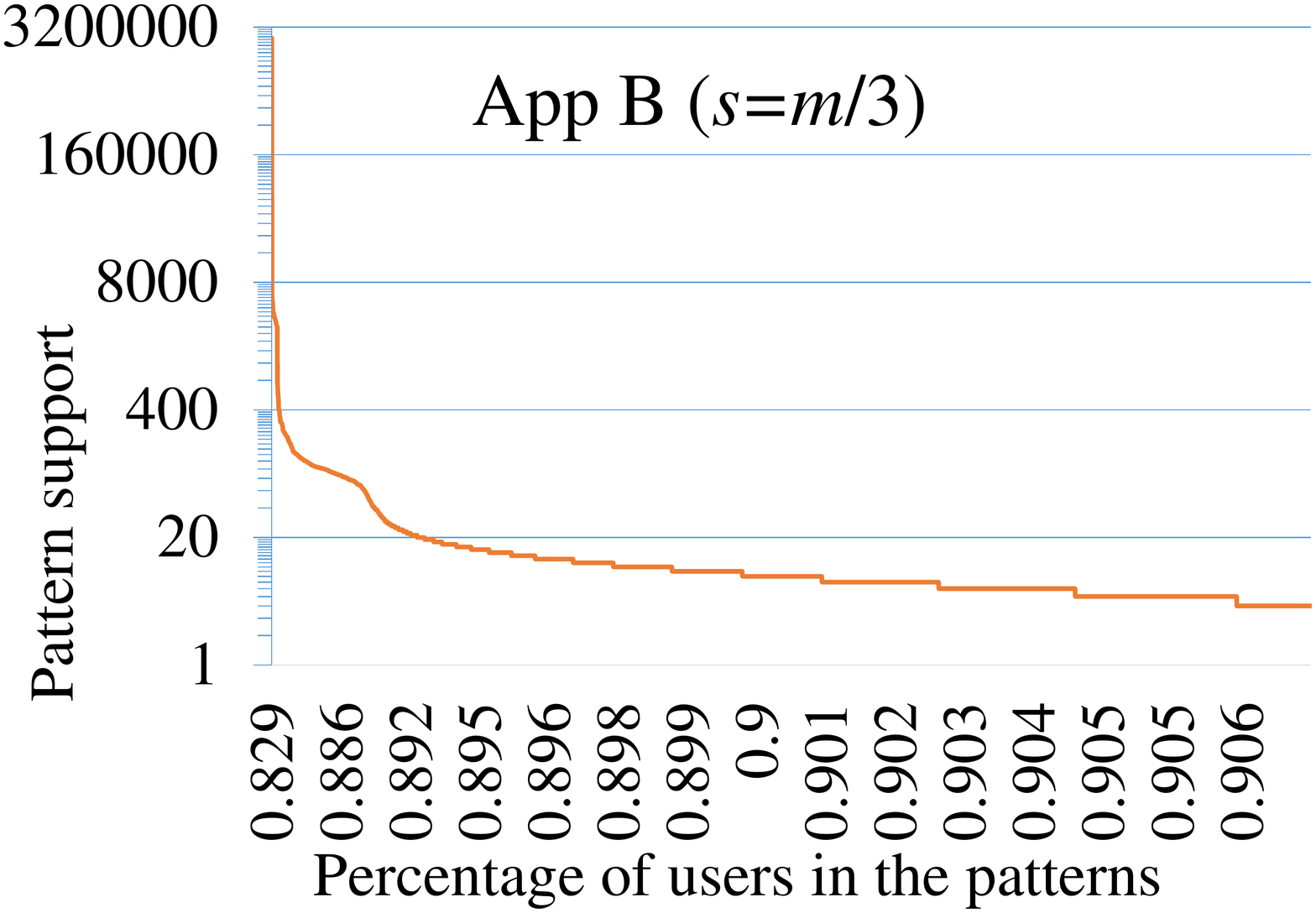}
	}
	\vspace{-0.5cm}
	\caption{Distribution of pattern supports for App A and B}
	\label{fig:patterns}
	\vspace{-0.5cm}
\end{figure}

It is not hard to see that theoretically for every behavior pattern there is a very large set of sequences of private counter values $\{x(t)\}_t$ that may map to it (depending on $\alpha_i$). Real data~(\cfig\ref{fig:patterns}) provides evidence that users tend to be approximately consistent and therefore simpler patterns, i.e., patterns that mostly stick to a single rounded value $y(t)=y$ correspond to larger sets of sequences $\{x_i(t)\}_t,$ obtained from a real population. In particular, for each app there is always one pattern (corresponding to having one fixed $y(t)=y$ across all 31 days) which blends the majority of users ($>2$ million). However more complex behavior patterns have less users mapping to them. In particular, there always are some lonely users ($1\%$-$5\%$ depending on $s$) who land in patterns that have support size of one or two. From the viewpoint of data collector such users can only be identified as those having a complex and irregular behavior, however the actual nature of that behavior by Theorem~\ref{Th:PermMem} remains uncertain.

\subsection{Example}\label{SubSec:AgeInDays}

One specific example of a counter collection problem that has been identified in~\cite[Section 1.3]{ccs:ErlingssonPK14} as being non-suitable for techniques presented in~\cite{ccs:ErlingssonPK14} but can be easily solved using our methods is to repeatedly collect age in days from a population of users. When we set $s=m$ and apply the algorithm of Section~\ref{SubSec:AlphaRound} we can collect such data for $T$ rounds with high accuracy. Each user necessarily responds with a sequence of bits that has form $z^t\circ {\bar z}^{T-t},$ where $0\leq t\leq T.$ Thus data collector only gets to learn the transition point, i.e., the day when user's age in days passes the value $m-\alpha_i,$ which is safe from privacy perspective as $\alpha_i$ is picked uniformly at random by the user.

\subsection{Continual collection for histogram estimation using permanent memoization}\label{SubSec:repeatedhistogram}

Since we discretize the range of values and map each user's value to a small number of $k$ buckets, $\alpha$-point rounding is not needed for histogram estimation. The single-round LDP mechanism in Duchi {\em et al.} \cite{nips:DuchiWJ13} sends out a 0-1 random response for each bucket: send $1$ with probability $e^{\eps/2}/(e^{\eps/2}+1)$ if the counter value is in this bucket, with probability $1/(e^{\eps/2}+1)$ if not. It is easy to see that this mechanism is $\epsilon$-LDP. Each user can memorize a mapping $f_k: [k] \rightarrow \{0,1\}^k$ by running this mechanism once for each $v \in [k]$, and always respond $f_k(v)$ if the users' value is in bucket $v$. However, this memoization schema leads to very serious privacy leakage. There is a situation where one has auxiliary information that can deterministically correlate a user's value with the output $z \in \{0,1\}^k$ produced by the algorithm: more concretely, if the data collector knows that the app usage value is in a bucket $v$ and observes the output $z$ in some day, whenever the user sends $z$ again in future, the data collector can infer that the bucket number is $v$ with almost 100\% probability.

To avoid such privacy leakages, we apply permanent memoization on our $d$-bit mechanism {\sf $d$BitFlip} (\csec\ref{sec:1bit:histogram}). Each user runs {\sf $d$BitFlip} once for each bucket number $v \in [k]$ and memoizes the response in a mapping $f_d: [k] \rightarrow \{0,1\}^d.$ The user will always send $f_d(v)$ if the bucket number is $v.$ This is mechanism is denoted by  {\sf $d$BitFlipPM}, and the same estimator \eqref{equ:histogramest} can be used to estimate the histogram upon receiving the $d$-bit response from every user. This scheme avoids several privacy leakages that arise due to memoization, because multiple ($\bigomega{k/2^d}$ w.h.p.) buckets are mapped to the same response. This protection is the strongest when $d = 1$. \cdef\ref{Def:BehPatt} about behavior patterns and \cthm\ref{Th:PermMem} can be naturally generalized here to provide similar privacy guarantee in repeated data collection.


\section{Output Perturbation}
\label{sec:opperturbation}

One of the limitations of memoization approach is that it does not protect the points of time where user's behavior changes significantly. Consider a user  who never uses an app for a long time, and then starts using it. When this happens, suppose the output produced by our algorithm changes from $0$ to $1.$ Then the data collector can learn with certainty that the user's behavior changed, (but not what this behavior was or what it became).  Output perturbation is one possible mechanism of protecting the {\em exact location} of the points of time where user's behavior has changed.  As mentioned earlier, output perturbation was introduced in \cite{ccs:ErlingssonPK14} as a way to mitigate privacy leakage that arises due to memoization. The main idea behind output perturbation is to flip the output of memoized responses with a small probability $ 0 \leq \gamma \leq 0.5$. This ensures that data collector will not be able to learn with certainty that behavior of a user changed at certain time stamps.

Consider the mean estimation algorithm. Suppose $b_i(t)$ denotes the memoized response bit for user $i$ at time $t.$ Then,
\begin{equation}\label{Eqn:OpAlg}
\hat b_i(t) =
\left\{
\begin{array}{ll}
b_i(t),& \mathrm{with \ probability\ } 1- \gamma; \\
1 - b_i(t) ,& \mathrm{otherwise}.                                                           \\
\end{array}
\right.
\end{equation}

Note that output perturbation is done at each time stamp $t$ on the memoized responses. To see how output perturbation protects users from the data collector learning exact points at which user's behavior changed, we need to set up some notation.
For an arbitrary $T > 0,$ fix a time horizon $[1, 2, \ldots T]$ where the counter data is collected. Let $x$ and $x'$ be two vectors in $[m]^T$, let $x(t)$ denote the $t$th coordinate of $x$ for $t \in [0,T]$.  Let $\mathcal{A}(x)$ and $\mathcal{A}(x')$ denote the output produced by our  1-bit algorithm + memoization. Let $\mathcal{A'}(x)$ and $\mathcal{A'}(x')$ denote the output produced by our  1-bit algorithm + memoization + output perturbation. Suppose the Hamming distance between $\mathcal{A}(x)$ and $\mathcal{A}(x')$ is at most $\delta$. Then,

\begin{theorem}\label{Th:OPLeakage}
Let $S$ be a vector in $\{0,1\}^T$. Then, $\frac{P[\mathcal{A'}(x) = S]}{P[\mathcal{A'}(x') = S]} \geq \gamma^{\delta}$.
\end {theorem}
Recall that in the output perturbation step, we flip each output bit $\mathcal{A}(x(t))$ independently with probability $\gamma$. This implies,
$$\frac{P[\mathcal{A'}(x) = S]}{P[\mathcal{A'}(x') = S]} = \Pi^{T}_{t = 1} \frac{P[\mathcal{A'}(x(t)) = S(t)]}{P[\mathcal{A'}(x'(t)) = S(t)]},$$

where we $S(t) \in \{0,1\}$ denotes the value of $S$ at the $t$th coordinate. For a $t \in [T]$ for which  $\mathcal{A}(x(t)) =  \mathcal{A}(x'(t))$, we have $\frac{P[\mathcal{A'}(x(t)) = S(t)]}{P[\mathcal{A'}(x'(t)) = S(t)]} =1$; this is true, since the probability used to flip the output bits is same for both the strings. Therefore,

\begin{equation}
\label{e:opratio}
\displaystyle \frac{P[\mathcal{A'}(x) = S]}{P[\mathcal{A'}(x') = S]} = \Pi_{t: t \in[T], \mathcal{A}(x(t))\not = \mathcal{A}(x'(t))} \frac{P[\mathcal{A'}(x(t)) = S(t)]}{P[\mathcal{A'}(x'(t)) = S(t)]},
\end{equation}

Now notice that for a $t \in [T]$ for which  $\mathcal{A}(x(t)) \not =  \mathcal{A}(x'(t))$, we have $\frac{P[\mathcal{A'}(x(t)) = S(t)]}{P[\mathcal{A'}(x'(t)) = S(t)]} \geq \gamma$. Thus, the lemma follows from Eq. (\ref{e:opratio}) and from our assumption that $|\{t: t \in[T], \mathcal{A}(x(t))\not = \mathcal{A}(x'(t)) \}| \leq \delta$.

The theorem implies that  if the user behavior changed at time $t$, then there is an interval of time $[t-\delta, t+\delta]$ where the data collector would not be able to differentiate if the user behavior changed at time $t$ or any other time $t' \in [t-\delta, t+\delta].$ Consider a user $i$ and let $x_i$ be a vector in $[m]^T$ that denotes the values taken by $i$ in the interval $[1, 2, ..., T]$. Suppose the user's behavior remains constant up to time step $t$, and it changes at time $t$, and then remains constant.  Without loss of generality, let us assume that $x_i(t') = a$ for all $t' < t$, and $x_i(t') = b$ for all $t' \geq t$.
Consider the case when the output produced by our memoization changes at time $t$; that is, using the notation from above paragraph, $\mathcal{A}(x_i(a)) \not = \mathcal{A}(x_i(b))$.  Without output perturbation, the data collector will be certain that user's value changed at time $t$. With output perturbation, we claim that the data collector would not be able to differentiate if the user's behavior changed at time $t$ or any other time $t' \in [t-\delta, t+\delta]$, if $\delta$ is sufficiently small. (Think of $\delta$ as some small constant.) We argue as follows. Consider another pattern of user's behavior $x'_i \in [m]^T$, $x'_i(t') = a$ for all $t' < t^*$ and $x'_i(t') = b$ for all $t' \geq t^*$. Further, if $t^* \in [t-\delta, t+\delta]$, then 	$\frac{P[\mathcal{A'}(x_i) = S]}{P[\mathcal{A'}(x'_i) = S]} \geq \gamma^{\delta}$. This is true because of the following reason. Consider the case $t^* \geq t$. Then, in the interval $[t, t + \delta]$, the output of 1-bit mechanism + memoization can be different for the strings $x_i, x'_i$. However, Hamming distance of $\mathcal{A}(x_i)$ and $\mathcal{A}(x'_i)$ is at most $\delta$. Thus, we conclude from Theorem \ref{Th:OPLeakage} that $\frac{P[\mathcal{A'}(x_i) = S]}{P[\mathcal{A'}(x'_i) = S]} \geq \gamma^{\delta}$. The argument for the case $t^* < t$ is exactly the same. Thus, output perturbation can help to protect learning exact points of time where the users' behavior changes.

Consider a single round of data collection with the algorithm above.
\begin{theorem}\label{Th:OP}
Using output perturbation with a positive $\gamma,$ in combination with the  $\epsilon$-DP {\sf 1BitMean} algorithm is equivalent to invoking the {\sf 1BitMean} algorithm with
\begin{equation}\label{Eqn:EPrime}
\epsilon^\prime=\ln \left (\frac{(1-2\gamma)(\frac{e^\epsilon}{e^\epsilon+1}) + \gamma}{(1-2\gamma)(\frac{1}{e^\epsilon+1}) + \gamma} \right).
\end{equation}
Thus, for each round of data collection, with probability at least $(1-\delta)$ the error of the mechanism presented above is at most $\left (m \cdot \frac{e^\epsilon+1}{(1-2\gamma) (e^\epsilon - 1)} \cdot \sqrt{\frac{1}{2n} \cdot \log \frac{2}{\delta}} \right),$ where $\delta$ is an arbitrary constant between zero and one.
\end{theorem}
\begin{proof}
Observe that the distribution produced by combining output perturbation in \eqref{Eqn:OpAlg} with the $\epsilon$-DP {\sf 1BitMean} algorithm in \eqref{equ:meanmech} is given by
\begin{equation}\label{Eqn:Alg}
\hat b_i(t) =
\left\{
\begin{array}{ll}
1,& \mathrm{with \ probability\ } (1-2\gamma)(\frac{1}{e^\epsilon+1}+\frac{x_i(t)}{m}\cdot \frac{e^\epsilon -1}{e^\epsilon+1}) + \gamma; \\
0,& \mathrm{otherwise}.                                                           \\
\end{array}
\right.
\end{equation}
It remains to note that if in formula~(\ref{equ:meanmech}) we use $\epsilon^\prime$ given by~(\ref{Eqn:EPrime}) instead of $\epsilon$, then~(\ref{equ:meanmech}) yields the same distribution as \eqref{Eqn:Alg}. Now to prove Theorem~\ref{Th:OP}, we simply invoke  Theorem \ref{thm:singlemean}.
\end{proof}

\eat{
}


\section{Empirical Evaluation}
\label{sec:exp}

We compare our mechanisms (with permanent memoization) for mean and histogram estimation with previous mechanisms for one-time data collection.
Note that all the mechanisms we compare here provide one-time $\eps$-LDP guarantee; however, our mechanisms provide additional protection for each individual's privacy during the repeated data collection (as introduced in \csecs\ref{Sec:RepeatedCollection}-\ref{sec:opperturbation}).
The goal of these experiments is to show that our mechanisms, with such additional protection, are no worse than or comparable to the state-of-the-art LDP mechanisms in terms of estimation accuracy.

We first use the real-world dataset which is described in Section \ref{SubSubSec:PatternDistribution}.

\begin{figure}[t]
	\center
	\subfigure[$n=0.3\times10^6$]{\label{fig:exp:mean:03}
		\includegraphics[width=0.316\textwidth]{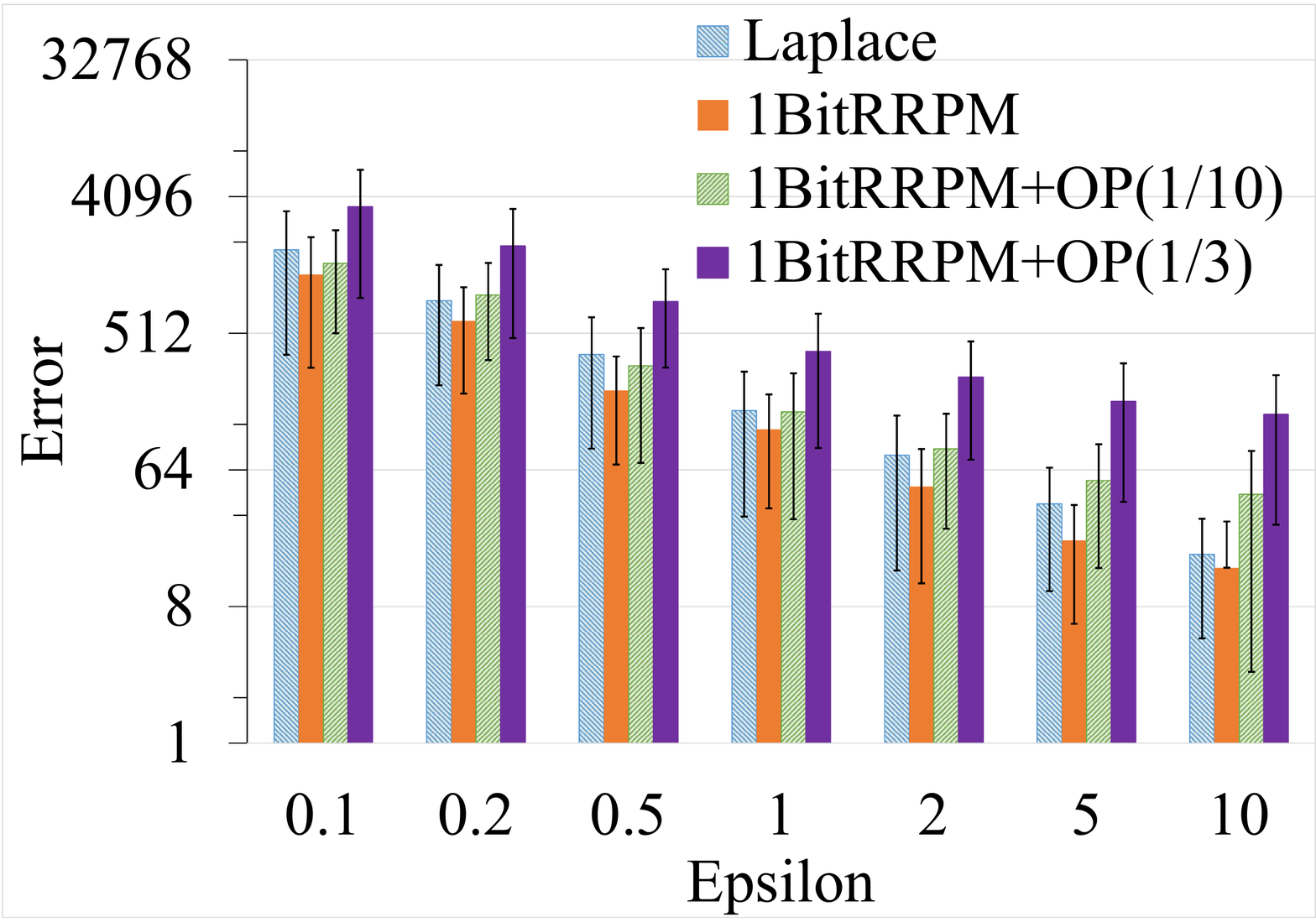}
	}
	\subfigure[$n=1\times10^6$]{
		\includegraphics[width=0.316\textwidth]{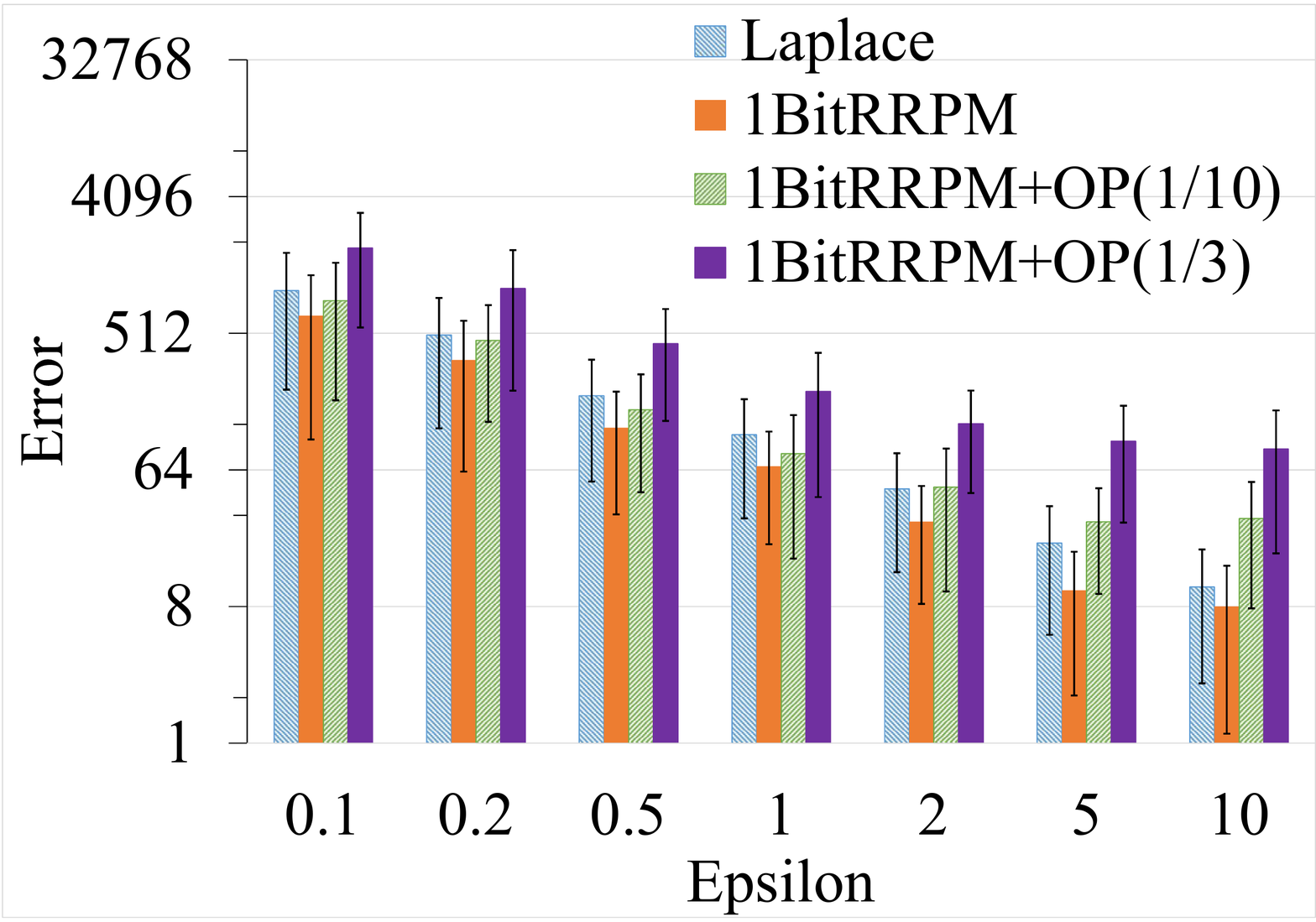}
	}
	\subfigure[$n=3\times10^6$]{
		\includegraphics[width=0.316\textwidth]{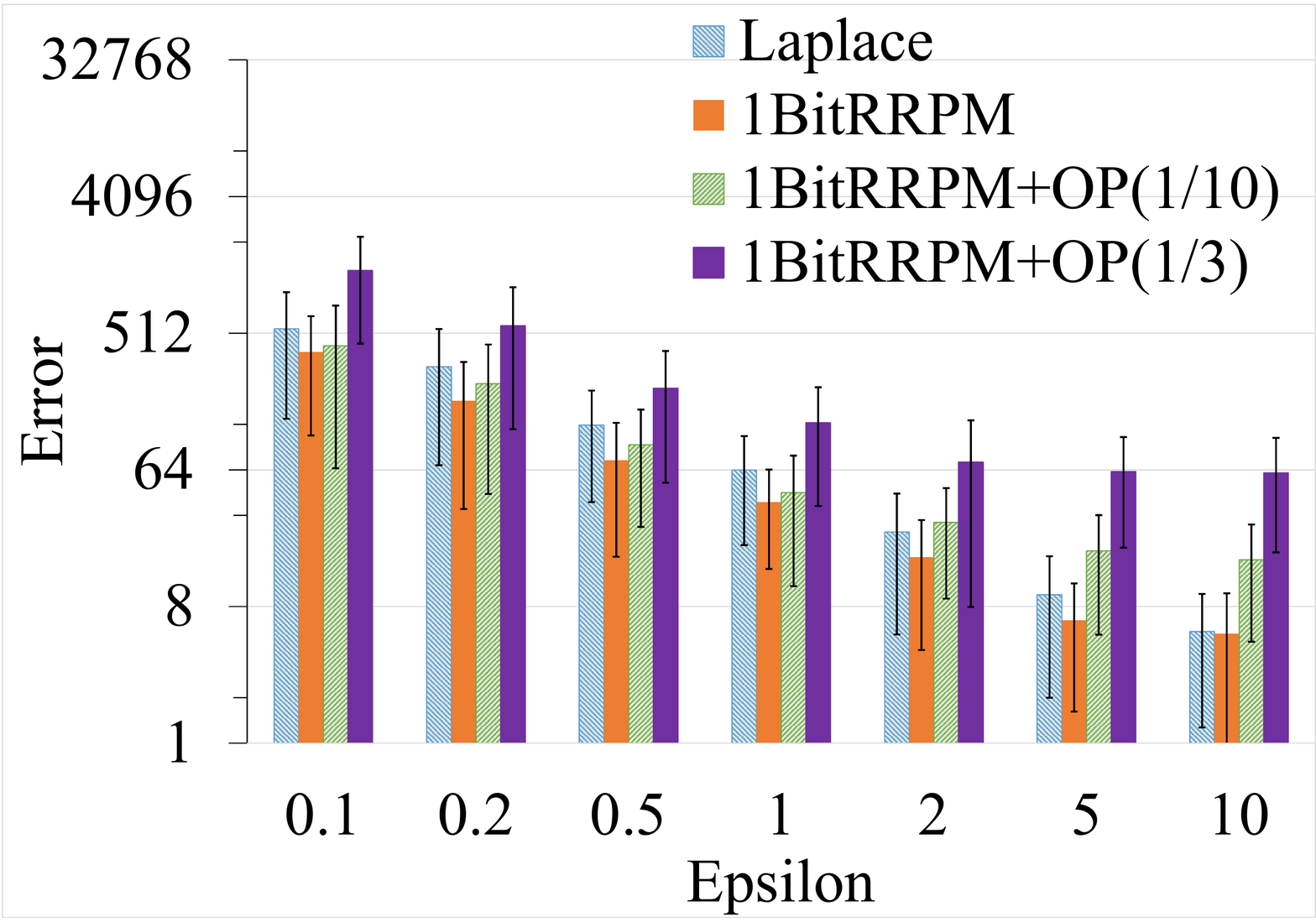}
	}
	\vspace{-0.5cm}
	\caption{Comparison of mechanisms for mean estimation (real-world datasets)}
	\label{fig:exp:mean}
	\vspace{-0.3cm}
\end{figure}

\begin{figure}[t]
	\center
	\subfigure[$n=0.3\times10^6$]{\label{fig:exp:histogram:03}
		\vspace{-0.3cm}
		\includegraphics[width=0.316\textwidth]{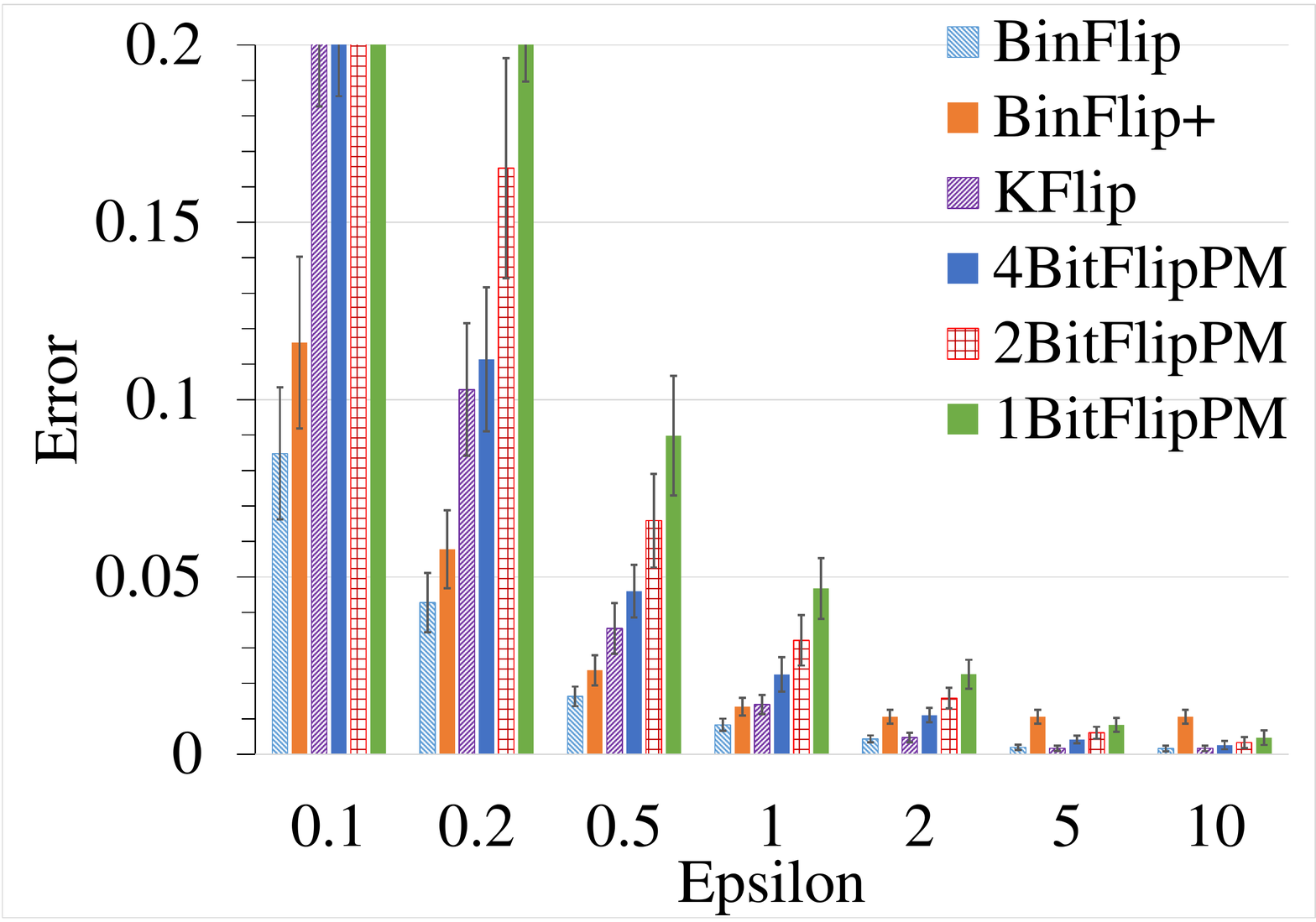}
	}
	\subfigure[$n=1\times10^6$]{
		\vspace{-0.3cm}
		\includegraphics[width=0.316\textwidth]{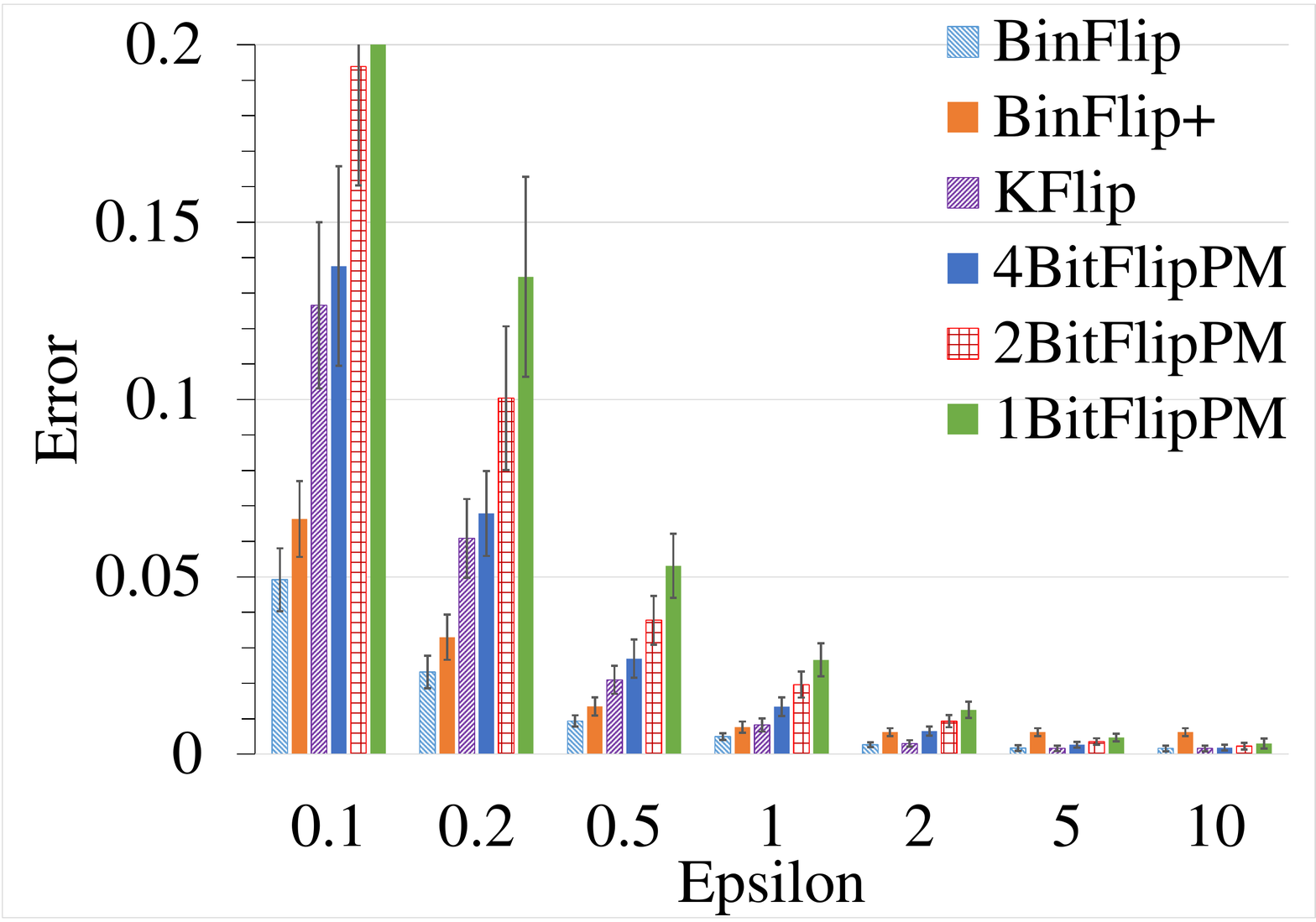}
	}
	\subfigure[$n=3\times10^6$]{
		\vspace{-0.3cm}
		\includegraphics[width=0.316\textwidth]{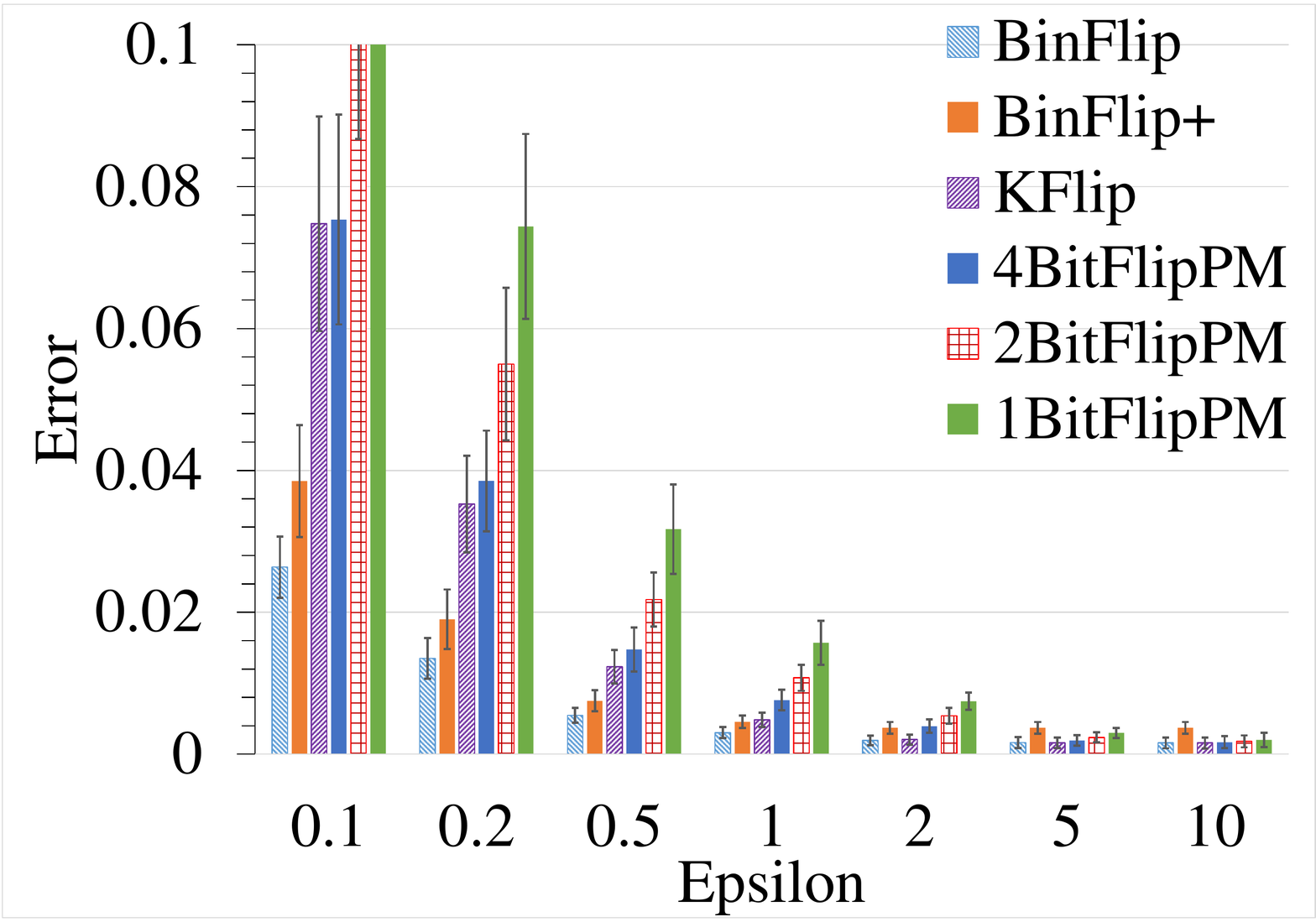}
	}
	\vspace{-0.4cm}
	\caption{Comparison of mechanisms for histogram estimation (real-world datasets)}
	\label{fig:exp:histogram}
	\vspace{-0.3cm}
\end{figure}

\stitle{Mean estimation.} We implement our 1-bit mechanism (introduced in \csec\ref{sec:1bit:mean}) with $\alpha$-point {\sf R}andomized {\sf R}ounding and {\sf P}ermanent {\sf M}emoization for repeated collection (\csec\ref{Sec:RepeatedCollection}), denoted by {\sf 1BitRRPM}, and output perturbation to enhance the protection for usage change (\csec\ref{sec:opperturbation}), denoted by {\sf 1BitRRPM+OP($\gamma$)}. We compare it with the Laplace mechanism for LDP mean estimation in \cite{nips:DuchiWJ13,corr:DuchiWJ16}, denoted by {\sf Laplace}. We vary the value of $\eps$ ($\eps=0.1$-$10$) and the number of users ($n=0.3, 1, 3\times10^6$ by randomly picking subsets of all the 3 million users), and run all the mechanisms 3000 times on the 31-day usage data with three counters. Recall that the domain size is $m=24~\text{hours}$. The average of absolute errors (in seconds) with one standard deviation (STD) are reported in \cfigs\ref{fig:exp:mean}. {\sf 1BitRRPM} is consistently better than {\sf Laplace} with smaller errors and narrower STDs. Even with a perturbation probability $\gamma=1/10$, they are comparable in accuracy. When $\gamma=1/3$, output perturbation is equivalent to adding an additional uniform noise from $[0, 24~{\rm hours}]$ independently on each day to provide very strong protection on usage change--even in this case, {\sf 1BitRRPM+OP(1/3)} gives us tolerable accuracy when the number of users is large.

\stitle{Histogram estimation.} We create $k=32$ buckets on $[0, 24~{\rm (hours)}]$ with even widths to evaluate mechanisms for histogram estimation. We implement our $d$-bit mechanism (\csec\ref{sec:1bit:histogram}) with permanent memoization for repeated collection (\csec\ref{SubSec:repeatedhistogram}), denoted by {\sf $d$BitFlipPM}. In order to provide protection on usage change in repeated collection, we use $d=1,2,4$ (strongest when $d=1$). We compare it with state-of-the-art one-time mechanisms for histogram estimation: {\sf BinFlip} \cite{nips:DuchiWJ13,corr:DuchiWJ16}, {\sf KFlip} \cite{kairouz2016discrete}, and {\sf BinFlip+} (applying the generic protocol with 1-bit reports in \cite{stoc:BassilyS15} on {\sf BinFlip}).  When $d=k$, {\sf $d$BitFlipPM} has the same accuracy as {\sf BinFlip}. {\sf KFlip} is sub-optimal for small $\eps$ \cite{kairouz2016discrete} but has better performance when $\eps$ is $\bigomega{\ln k}$. In contrast, {\sf BinFlip+} has good performance when $\eps \leq 2$.

We repeat the experiment 3000 times and report the average {\em histogram error} (\ie, maximum error across all bars in a histogram) with one standard deviation for different algorithms in \cfig\ref{fig:exp:histogram} with $\eps = 0.1$-$10$ and $n=0.3, 1, 3 \times10^6$ to confirm the above theoretical results. {\sf BinFlip} (equivalently, {\sf 32BitFlipPM}) has the best accuracy overall.

With enhanced privacy protection in repeated data collection, {\sf 4bitFlipPM} is comparable to the one-time collection mechanism {\sf KFlip} when $\eps$ is small ($0.1$-$0.5$); and {\sf 4bitFlipPM}-{\sf 1bitFlipPM} are better than {\sf BinFlip+} when $\eps$ is large ($5$-$10$).

\stitle{On different data distributions.} We have shown that errors in mean and histogram estimations can be bounded (\cthms\ref{thm:singlemean}-\ref{thm:singlehistogram}) in terms of $\eps$ and the number of users $n$, together with the number of buckets $k$ and the number of bits $d$ (applicable only to histograms). We now conduct additional experiments on synthetic datasets to verify that the empirical errors should not change much on different data distributions. Three types of distributions are considered: i) constant distribution, \ie, each user $i$ has a counter $x_i(t) = 12~{\rm (hours)}$ all the time; ii) uniform distribution, \ie, $x_i(t) \sim {\cal U}(0, 24)$; and iii) normal distribution, \ie, $x_i(t) \sim {\cal N}(12, 2^2)$ (with mean equal to $12$ and standard deviation equal to $2$), truncated on $[0, 24]$. Three synthetic datasets are created by drawing samples of sizes $n = 0.3 \times 10^6$ from these three distributions. Results are plotted on \cfigs\ref{fig:exp:mean:synthetic}-\ref{fig:exp:histogram:synthetic} for mean and histogram estimations, respectively, and are almost the same as those in \cfigs\ref{fig:exp:mean:03} and \ref{fig:exp:histogram:03}.

\begin{figure}[t]
	\center
	\subfigure[Constant]{
		\includegraphics[width=0.316\textwidth]{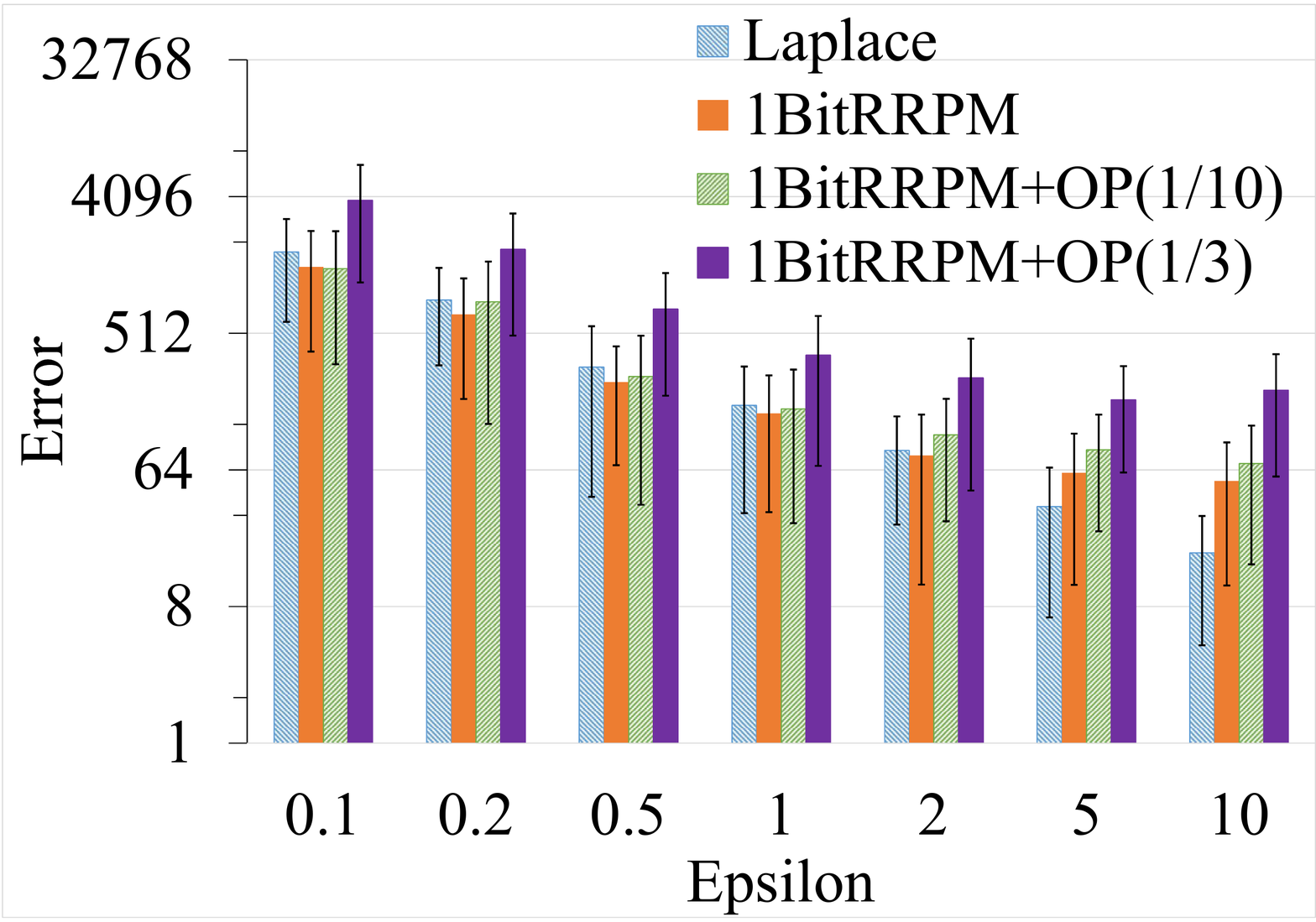}
	}
	\subfigure[Uniform distribution]{
		\includegraphics[width=0.316\textwidth]{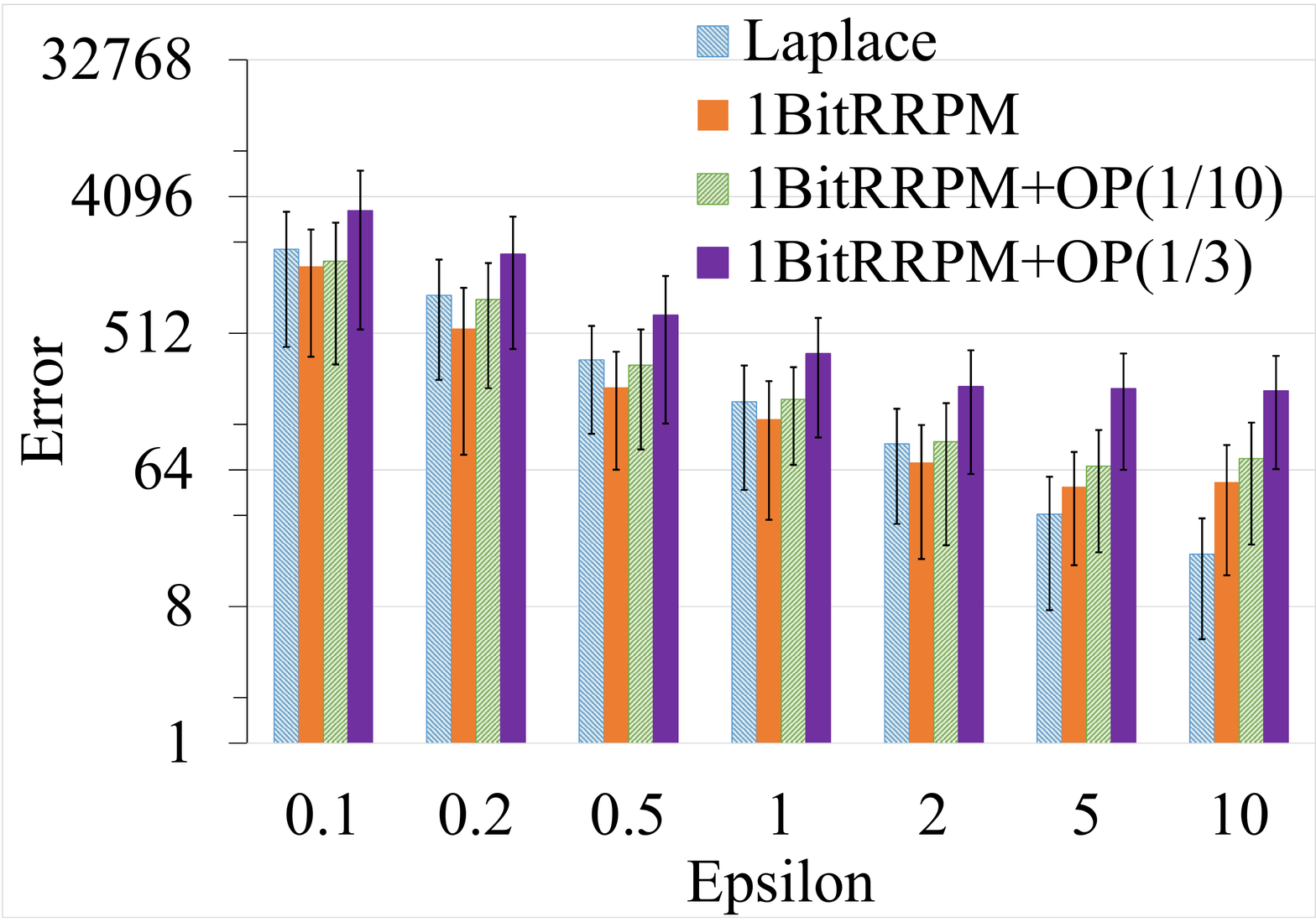}
	}
	\subfigure[Normal distribution (truncated)]{
		\includegraphics[width=0.316\textwidth]{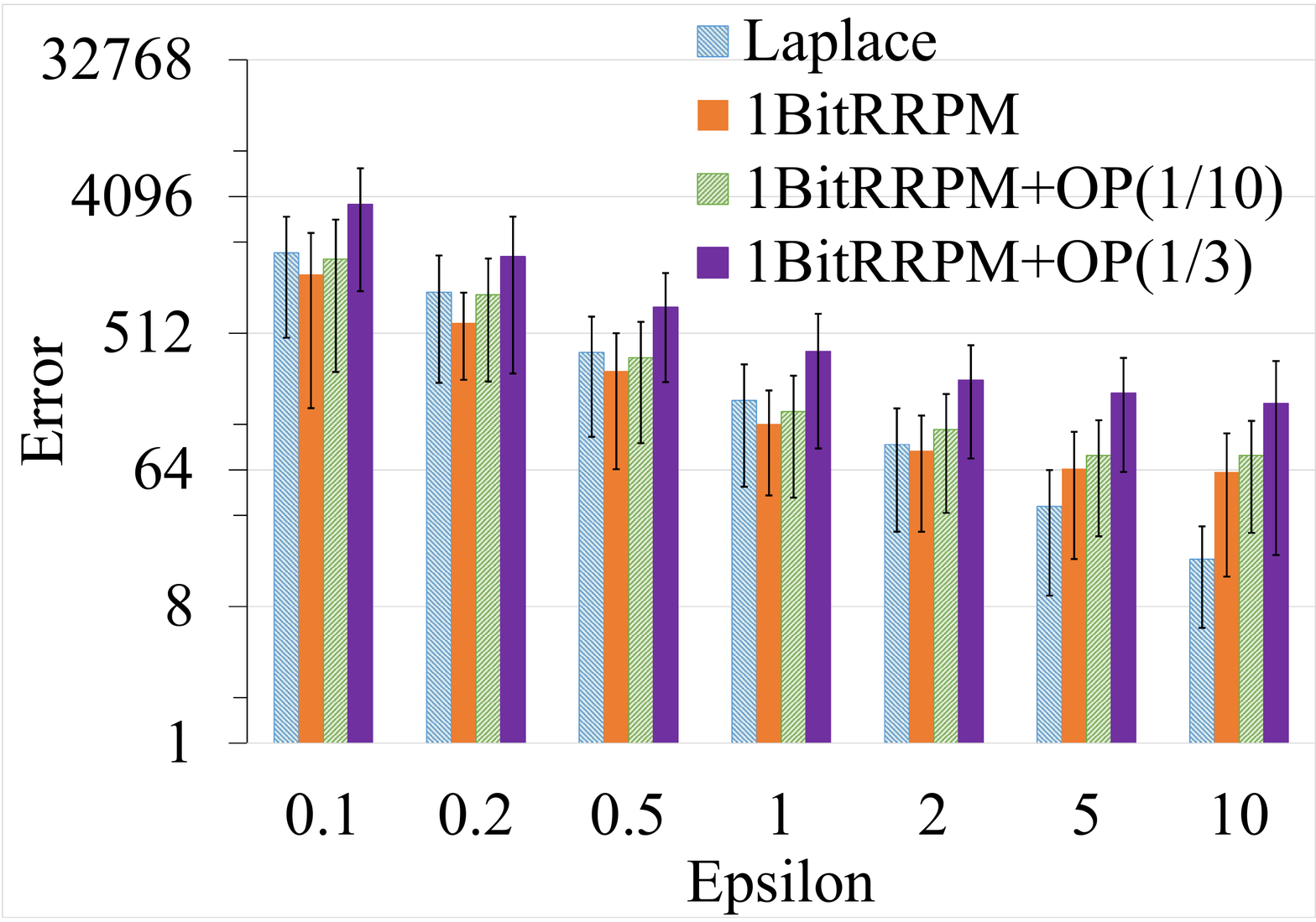}
	}
	\vspace{-0.5cm}
	\caption{Comparison of mechanisms for mean estimation (synthetic datasets: $n=0.3\times10^6$)}
	\label{fig:exp:mean:synthetic}
	\vspace{-0.3cm}
\end{figure}

\begin{figure}[t]
	\center
	\subfigure[Constant distribution]{
		\vspace{-0.3cm}
		\includegraphics[width=0.316\textwidth]{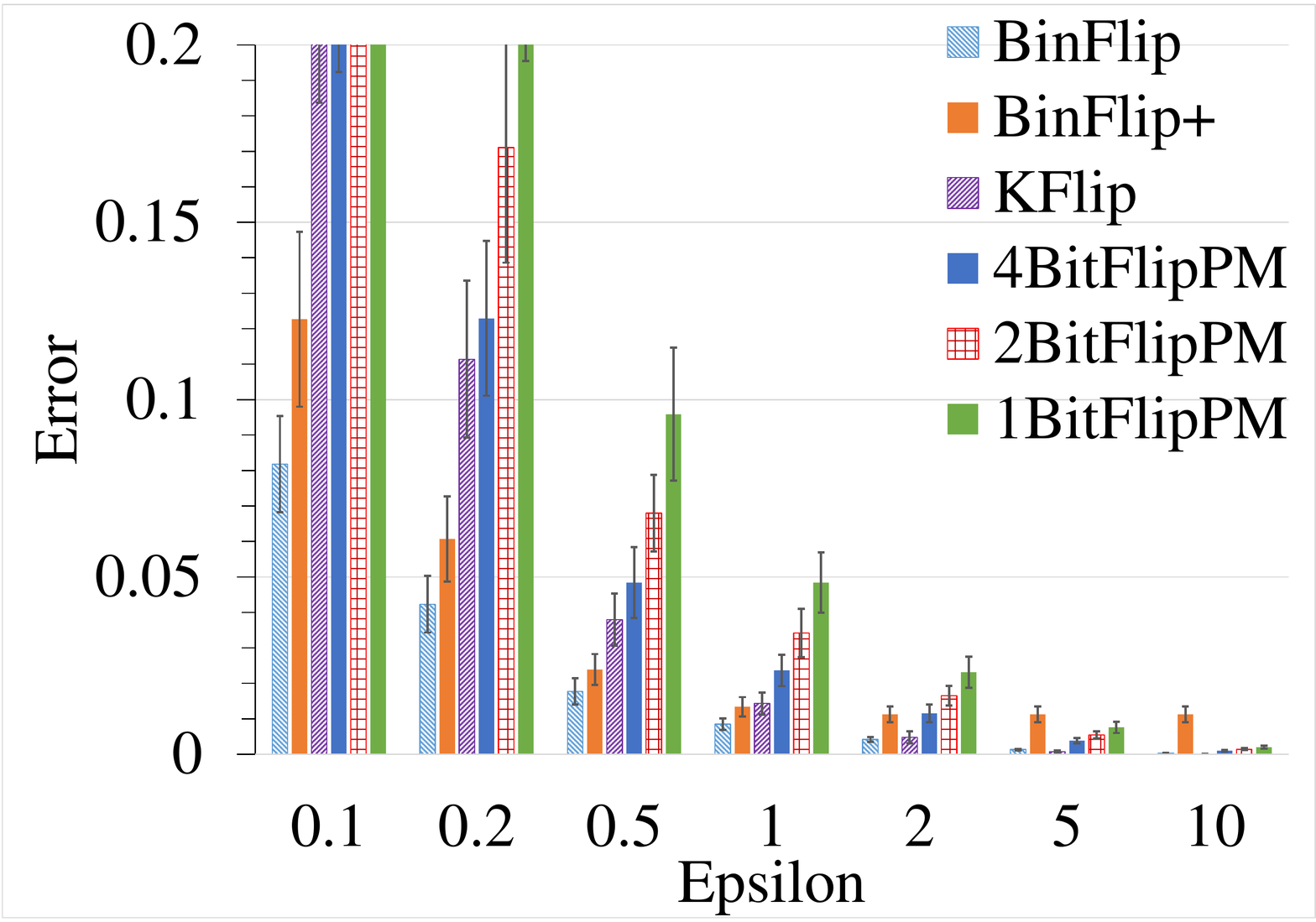}
	}
	\subfigure[Uniform distribution]{
		\vspace{-0.3cm}
		\includegraphics[width=0.316\textwidth]{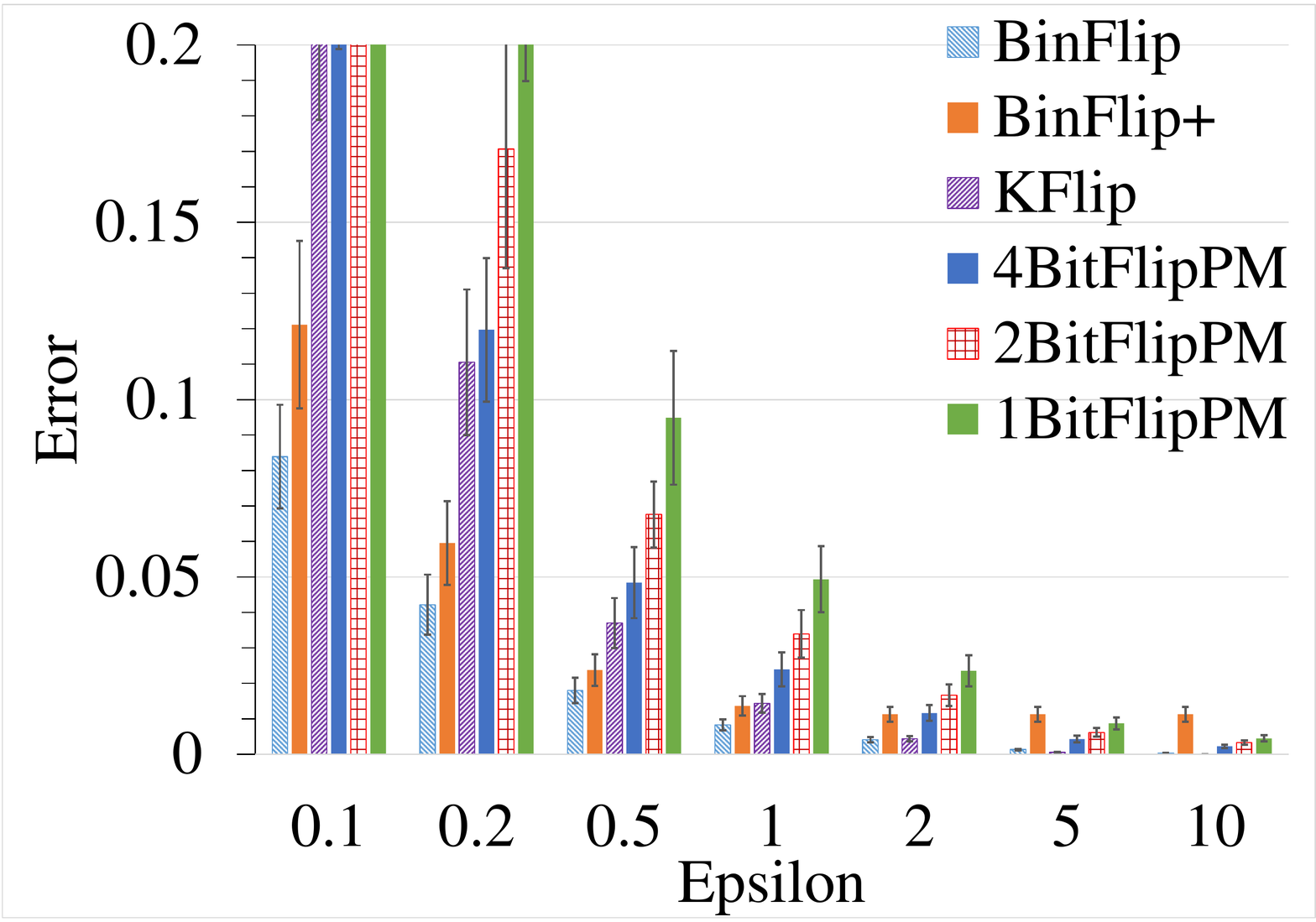}
	}
	\subfigure[Normal distribution (truncated)]{
		\vspace{-0.3cm}
		\includegraphics[width=0.316\textwidth]{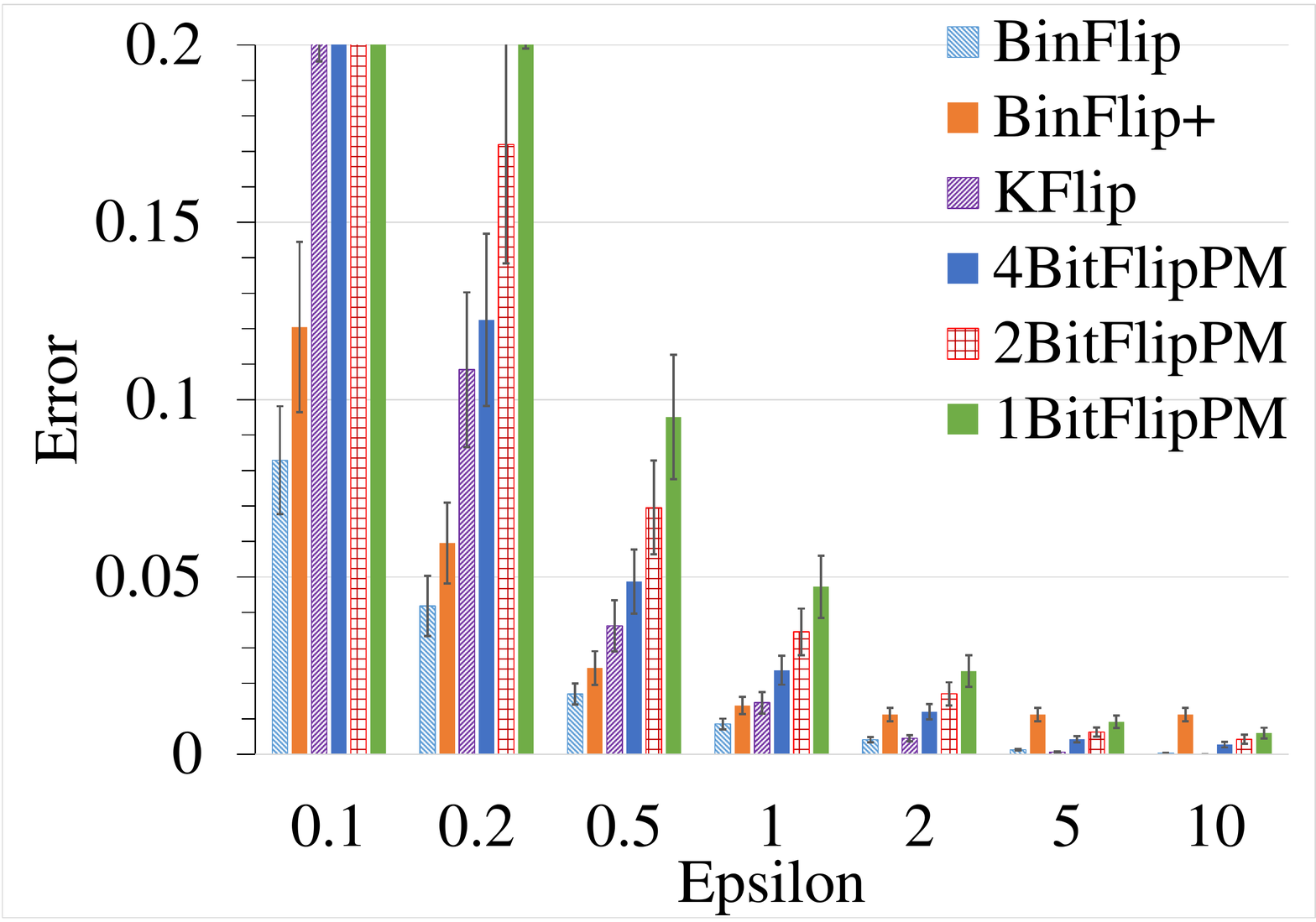}
	}
	\vspace{-0.4cm}
	\caption{Comparison of mechanisms for histogram estimation (synthetic datasets: $n=0.3\times10^6$)}
	\label{fig:exp:histogram:synthetic}
	\vspace{-0.3cm}
\end{figure}

\section{Deployment}\label{Sec:Deployment}
In earlier sections, we presented new LDP mechanisms geared towards repeated collection of counter data, with formal privacy guarantees even after being executed for a long period of time. Our mean estimation algorithm has been deployed by Microsoft starting with Windows Insiders in Windows 10 Fall Creators Update. The algorithm is used to collect the number of seconds that a user has spent using a particular app. Data collection is performed every 6 hours, with $\epsilon=1.$ Memoization is applied across days and output perturbation uses $\gamma=0.2.$ According to Theorem~\ref{Th:OP}, this makes a single round of data collection satisfy $\epsilon^\prime$-DP with $\epsilon^\prime=0.686.$

One important feature of our deployment is that collecting usage data for multiple apps from a single user only leads to a minor additional privacy loss that is independent of the actual number of apps. Intuitively, this happens since we are collecting active usage data, and the total number of seconds that a user can spend across multiple apps in 6 hours is bounded by an absolute constant that is independent of the number of apps.
\begin{theorem}\label{Th:MultiApp}
Using the {\sf 1BitMean} mechanism with a privacy parameter $\tau$ to simultaneously collect $t$ counters $x_1,\ldots,x_t,$ where each $x_i$ satisfies $0\leq x_i\leq m$ and $\sum_i x_i\leq m$ preserves $\tau^{\prime}$-DP, where
\begin{equation}\label{Eqn:TauPrime}
\tau^\prime = \tau+e^{\tau}-1.
\end{equation}
\end{theorem}
\begin{proof}
For $z\in \{0,1\}$ and an integer $0\leq x\leq m,$ let $\mathrm{p}(z\mid x)$ denote the probability that the {\sf 1BitMean} mechanism produces an output $z$ on an input $x,$ as given by~(\ref{equ:meanmech}). Let $x_1,\ldots,x_t$ and $y_1,\ldots,y_t$ be two sets of arbitrary counter values. Here all $\{x_i\}_{i\in [t]}$ and $\{y_i\}_{i\in [t]}$ are non-zero, $\sum_i{x_i}\leq m,$ and $\sum_i{y_i}\leq m.$ Fix some $z_1,\ldots,z_t\in \{0,1\}^t.$
We need to bound
\begin{equation}
R=\prod_{i\in [t]}\frac{\mathrm{p}(z_i\mid x_i)}{\mathrm{p}(z_i\mid y_i)}.
\end{equation}
Let $S=\{i\in [t]\mid z_i=0\}.$ We have
\begin{equation}\label{Eqn:RFull}
R=\prod_{i\in S}     \frac{\mathrm{p}(0\mid x_i)}{\mathrm{p}(0\mid y_i)} \cdot
  \prod_{i\in \bar S}\frac{\mathrm{p}(1\mid x_i)}{\mathrm{p}(1\mid y_i)}.
\end{equation}
Note that $\mathrm{p}(0\mid x_i)\leq \mathrm{p}(0\mid 0).$ Thus using formula~(\ref{equ:meanmech}),
\begin{eqnarray*}
  \prod_{i\in S} \frac{\mathrm{p}(0\mid x_i)}{\mathrm{p}(0\mid y_i)} & \leq &
  \prod_{i\in S} \left[\frac{e^\tau}{e^\tau+1} \cdot \frac{1}{\mathrm{p}(0\mid y_i)} \right] \nonumber \\
  & = &
  \left(\prod_{i\in S}\left[1-\frac{y_i}{m}\cdot \frac{e^\tau-1}{e^\tau}\right]\right)^{-1}. \nonumber \\
\end{eqnarray*}
It remains to note that
\begin{equation}\label{Eqn:R01}
\prod_{i\in S}\left[1-\frac{y_i}{m}\cdot \frac{e^\tau-1}{e^\tau}\right] \geq \frac{1}{e^\tau},
\end{equation}
as the product above is minimized when one of $y_i$ is set to $m$ and the rest are zero. Therefore
\begin{equation}\label{Eqn:R02}
\prod_{i\in S} \frac{\mathrm{p}(0\mid x_i)}{\mathrm{p}(0\mid y_i)} \leq e^\tau.
\end{equation}
We proceed to bound the second product in~(\ref{Eqn:RFull}). Since $\mathrm{p}(1\mid y_i)\geq \mathrm{p}(1\mid 0),$ by~(\ref{equ:meanmech})we have,
\begin{eqnarray*}
  \prod_{i\in \bar S} \frac{\mathrm{p}(1\mid x_i)}{\mathrm{p}(1\mid y_i)} & \leq &
  \prod_{i\in \bar S} \left[(e^\tau+1) \cdot \mathrm{p}(1\mid x_i) \right] \nonumber \\
  & = &
  \prod_{i\in \bar S}\left[ 1+\frac{x_i}{m}\cdot (e^\tau-1) \right]  \nonumber \\
  & \leq &
  \prod_{i\in \bar S}\left[ 1+\frac{1}{|\bar S|}\cdot (e^\tau-1) \right]  \nonumber \\
  & \leq &
  e^{e^\tau-1}. \nonumber \\
\end{eqnarray*}
Combining~(\ref{Eqn:R02}) and the inequality above, we conclude that
\begin{equation}\label{Eqn:RLast}
R\leq e^{\tau}\cdot e^{e^\tau-1},
\end{equation}
which concludes the proof.
\end{proof}
By Theorem~\ref{Th:MultiApp}, in deployment, a single round of data collection across an arbitrary large number of apps satisfies $\epsilon^{\prime\prime}$-DP, where $\epsilon^{\prime\prime}=1.672.$

\bibliographystyle{abbrv}
\bibliography{ref}

\end{document}